\documentclass[10pt,journal,compsoc]{IEEEtran}
\usepackage[framemethod=TikZ]{mdframed}
\mdfsetup{nobreak=false}
\newenvironment{Assumption}[1][]{\ifstrempty{#1}{\mdfsetup{frametitle={\tikz[baseline=(current bounding box.east),outer sep=0pt]
      \node[line width=0pt,anchor=east,rectangle,draw=white,fill=white]
    ;}}
  }{\mdfsetup{frametitle={\tikz[baseline=(current bounding box.east),outer sep=0pt]
      \node[anchor=east,rectangle,draw=white,fill=white]
    {\strut #1};}}}\mdfsetup{innertopmargin=-5pt,linecolor=black,linewidth=0.5pt,topline=true,frametitleaboveskip=\dimexpr-\ht\strutbox\relax,skipabove=\topskip,skipbelow=\topskip}
  \begin{mdframed}[nobreak=false]\relax
  }{\end{mdframed}}

\usepackage{xspace}
\usepackage{colortbl}
\usepackage{amsmath,amsfonts,amssymb}
\usepackage{graphicx}
\usepackage{textcomp}
\usepackage{xcolor}
\usepackage{tikz}
\usepackage[most]{tcolorbox}
\usepackage{graphbox}
\usepackage{mathtools}
\usepackage{fancybox}
\usepackage{makecell}
\usepackage{multirow}
\usepackage{booktabs}
\usepackage{hyperref}
\hypersetup{hidelinks}
\usepackage{cleveref} 
\usepackage[normalem]{ulem}
\usepackage{listings}
\usepackage{caption}
\usepackage{subcaption}
\usepackage{longtable}
\usepackage{pifont}
\usepackage{ifthen}
\usepackage[english]{babel} 
\usepackage{amsthm}
\usepackage{enumitem}
\usepackage{stmaryrd}
\usepackage{subfloat}
\usepackage[latin1]{inputenx}
\usepackage[T1]{fontenc}
\usepackage{soul}
\usepackage{algorithm}
\usepackage{algpseudocode}

\definecolor{keywordcolor}{RGB}{127,0,85}
\definecolor{background}{rgb}{0.94,0.95,0.96}
\definecolor{preconditionColor}{HTML}{e6ffe0}
\definecolor{postconditionColor}{HTML}{ffffcc}

\newcommand{\lit}[1]{\textbf{\texttt{\textcolor{keywordcolor}{#1}}}}
\newcommand{\prev}{\lit{\ensuremath{\texttt{prev}}}\xspace}
\newcommand{\duration}{\lit{\ensuremath{\texttt{dur}}}\xspace}

\newcommand\synt[1]{\texttt{#1}}

\newcommand{\timevariableEncoding}{\ensuremath{\tau}\xspace}

\newcommand{\preconditions}{\synt{pre}\xspace}

\newcommand{\domain}{\ensuremath{\mathbb{D}}\xspace}

\newcommand{\variables}{\ensuremath{\mathcal{V}}\xspace}
\newcommand{\inputs}{\ensuremath{\mathcal{U}}\xspace}
\newcommand{\outputs}{\ensuremath{\mathcal{Y}}\xspace}

\newcommand{\type}{\ensuremath{\mathtt{dom}}\xspace}

\newcommand{\inputVariable}{\ensuremath{u}\xspace}
\newcommand{\outputVariable}{\ensuremath{y}\xspace}
\newcommand{\variable}{\ensuremath{v}\xspace}
\newcommand{\term}{\synt{a}\xspace}
\newcommand{\logicalexpression}{\synt{e}\xspace}

\newcommand{\constant}{\ensuremath{c}\xspace}

\newcommand{\trace}{\ensuremath{\langle \inputInterpretation, \outputInterpretation, \interpretationtime \rangle}\xspace}
\newcommand{\tracesymbol}{\ensuremath{\pi}\xspace}
\newcommand{\traces}{\ensuremath{\Pi}\xspace}

\newcommand\interpretation[1]{\ensuremath{\llbracket #1 \rrbracket}}
\newcommand{\codespace}{\hspace{1cm}}
\newcommand{\RTtoStateflow}{\textsc{RT2Stateflow}\xspace}
\newcommand{\RT}{\textsc{RT}\xspace}

\newcommand{\nonnegativenatural}{\ensuremath{\mathbb{N}^0}\xspace}

\newcommand{\real}{\ensuremath{\mathbb{R}}\xspace}
\newcommand{\positivenatural}{\ensuremath{\mathbb{N}^+}\xspace}

\newcommand{\durationcolumn}{\synt{d}\xspace}
\newcommand{\postconditions}{\synt{post}\xspace}
\newcommand{\inputInterpretation}{\ensuremath{\iota_\inputs}\xspace}
\newcommand{\outputInterpretation}{\ensuremath{\iota_\outputs}\xspace}
\newcommand{\indexvariable}{\ensuremath{\texttt{i}}\xspace}
\newcommand{\interpretationtime}{\ensuremath{\iota_\tau}\xspace}

\newcommand{\requirementTable}{\synt{rt}\xspace}
\newcommand{\sampleTime}{\ensuremath{\textsc{T}_\textsc{s}}\xspace}

\newcommand{\requirement}{\synt{req}\xspace}

\newcommand{\simulink}{Simulink\textsuperscript{\tiny\textregistered}\xspace}

\newcommand{\SLTest}{Simulink\textsuperscript{\tiny\textregistered}
Test\textsuperscript{\tiny TM}\xspace}

\newcommand{\matlab}{MATLAB\textsuperscript{\tiny\textregistered}\xspace}

\newcommand{\failingtestsequence}{\textsc{TC}}
\newcommand{\nff}{\textsc{NFF}\xspace}
\newcommand{\cmark}{\ding{51}}  \newcommand{\xmark}{\ding{55}}  \newcommand{\ReqTbx}{Requirements Toolbox\textsuperscript{\tiny TM}\xspace}

\newcommand\modelrtcombinations{\ensuremath{60}\xspace}
\newcommand\rqonetotal{\ensuremath{85\%}\xspace}

\newboolean{showcomments}
\setboolean{showcomments}{false}
\ifthenelse{\boolean{showcomments}}
{\newcommand{\nb}[2]{
  \fcolorbox{black}{yellow}{\bfseries\sffamily\scriptsize#1}
  {\sf\small$\blacktriangleright$\textit{#2}$\blacktriangleleft$}
 }
 
}
{\newcommand{\nb}[2]{}
 
}

\usepackage[framemethod=TikZ]{mdframed}

\newcommand\phase[1]{\tikz[baseline=(X.base)]\node [draw=black,fill=white,thick,rectangle,inner sep=2pt, rounded corners=2pt](X){\color{black}\textbf{#1}};}

\def\BibTeX{{\rm B\kern-.05em{\sc i\kern-.025em b}\kern-.08em
    T\kern-.1667em\lower.7ex\hbox{E}\kern-.125emX}}

\makeatletter
\renewcommand{\Function}[2]{\csname ALG@cmd@\ALG@L @Function\endcsname{#1}{#2}\def\jayden@currentfunction{#1}}
\newcommand{\funclabel}[1]{\@bsphack
  \protected@write\@auxout{}{\string\newlabel{#1}{{\jayden@currentfunction}{\thepage}}}\@esphack
}
\makeatother

\theoremstyle{plain}
\newtheorem{theorem}{Theorem}[section]

\theoremstyle{definition}

\newtheorem{definition}[theorem]{Definition}

\crefname{definition}{definition}{definitions}
\Crefname{definition}{Definition}{Definitions}

\theoremstyle{plain}

\makeatletter
\def\ps@IEEEtitlepagestyle{
    \def\@oddfoot{\mycopyrightnotice}
    \def\@evenfoot{}
}
\def\mycopyrightnotice{
    {\footnotesize
        \begin{minipage}{\textwidth}
            \centering
            \textcopyright~{\it ``This work has been submitted to the IEEE for possible publication. Copyright may be transferred without notice, after which this version may no longer be accessible.''}
        \end{minipage}
    }
}

\begin{document}

\title{Search-based Testing of Simulink Models with Requirements Tables}

\author{Federico~Formica,
Chris~George,
Shayda~Rahmatyan,
Vera~Pantelic,
Mark~Lawford,
Angelo~Gargantini,
Claudio~Menghi
\thanks{
F.~Formica, C.~George, S.~Rahmatyan, V.~Pantelic, and M.~Lawford are with the McMaster University, Hamilton, Canada -
e-mail: \{formicaf, georgc9, rahmats, pantelv, lawford\}@mcmaster.ca\newline
A.~Gargantini is with University of Bergamo, Bergamo, Italy -
email: angelo.gargantini@unibg.it\newline
C. Menghi is with University of Bergamo, Bergamo, Italy and McMaster University, Hamilton, Canada -
e-mail: claudio.menghi@unibg.it}
}

\begin{comment}
\author{Federico Formica}
\orcid{0000-0002-3033-7371}
\affiliation{\institution{McMaster University}
  \streetaddress{}
  \city{Hamilton}
  \country{Canada}}
\email{formicaf@mcmaster.ca}

\author{Chris George}
\orcid{0009-0000-7371-3763}
\affiliation{\institution{McMaster University}
  \streetaddress{}
  \city{Hamilton}
  \country{Canada}}
\email{georgc9@mcmaster.ca}

\author{Shayda Rahmatyan}
\orcid{0009-0000-4763-5335}
\affiliation{\institution{McMaster University}
  \streetaddress{}
  \city{Hamilton}
  \country{Canada}}
\email{rahmats@mcmaster.ca}

\author{Vera Pantelic}
\orcid{0000-0003-1696-2768}
\affiliation{\institution{McMaster University}
  \streetaddress{}
  \city{Hamilton}
  \country{Canada}}
\email{pantelv@mcmaster.ca}

\author{Mark Lawford}
\orcid{0000-0003-3161-2176}
\affiliation{\institution{McMaster University}
  \streetaddress{}
  \city{Hamilton}
  \country{Canada}}
\email{lawford@mcmaster.ca}

\author{Angelo Gargantini}
\orcid{0000-0002-4035-0131}
\affiliation{\institution{University of Bergamo}
  \streetaddress{}
  \city{Bergamo}
  \country{Italy}}

\author{Claudio Menghi}
\orcid{0000-0001-5303-8481}
\affiliation{\institution{University of Bergamo}
  \streetaddress{}
  \city{Bergamo}
  \country{Italy}}
\affiliation{\institution{McMaster University}
  \streetaddress{}
  \city{Hamilton}
  \country{Canada}}
\email{claudio.menghi@unibg.it}
\end{comment}

\maketitle              
%\linenumbers

\begin{abstract}
Search-based software testing (SBST) of Simulink models helps to find scenarios that demonstrate that the system can reach a state that violates one of its requirements.
However, many SBST techniques for \simulink models rely on requirements expressed in logical languages often not supported in industrial tools.
Integrating SBST methods with the tools already existing within Simulink for specifying requirements can increase the methods' practical adoption. 
This work presents the first black-box approach for testing Simulink models that supports Requirements Table (RT), a tool from the Simulink Requirements Toolbox used by practitioners to specify software requirements.

We evaluated our solution by considering \modelrtcombinations model-RT combinations consisting of a model and an RT. 
Our SBST framework returned a failure-revealing test case for \rqonetotal of the model-RT combinations.
Remarkably, it identified a failure-revealing test case for three model-RT combinations for a cruise controller of an industrial simulator that other previously used tools could not find.
The efficiency of our SBST solution is acceptable for practical applications and comparable to existing SBST tools that are not based on RT.
\end{abstract}

\begin{IEEEkeywords}
Testing, Falsification, Simulink, Requirements Table.
\end{IEEEkeywords}

\newcounter{commentnumber}

\newcommand\rwcomment[2]{~\\ \noindent{\bf   \refstepcounter{commentnumber}\label{#1}\comm{C\thecommentnumber} }   \emph{#2}}

\newcommand\response[1]{~\\\noindent \textbf{Response:} #1}

\newcommand\referenceaddressed[1]{\ifthenelse{\equal{\ref{#1}}{\ref{#1end}}}{page~\pageref{#1} line~\ref{#1}}{page~\pageref{#1} lines~\ref{#1}---\ref{#1end}}}

\newcommand\addcitation[3]{\linelabel{#2} \textcolor{blue}{\cite{#1}} \comm{C\ref{#3}} \linelabel{#2end}  }

\newcommand\deletecitation[3]{\linelabel{#2} \textcolor{blue}{\sout{\mbox{\cite{#1}}}} \comm{C\ref{#3}} \linelabel{#2end}  }

\newcommand\delete[3]{\linelabel{#2}\textcolor{blue}{\sout{#1}} \comm{C\ref{#3}} \linelabel{#2end}  }

\newcommand\change[3]{\linelabel{#2}\textcolor{blue}{#1} \comm{C\ref{#3}} \linelabel{#2end}  }

\newcommand\changeRead[3]{\linelabel{#2}\textcolor{gray}{#1} \comm{C\ref{#3}} \linelabel{#2end}  }

\newcommand\rep[4]{\linelabel{#3}\textcolor{blue}{#1} \textcolor{blue}{\sout{#2}} \comm{C\ref{#4}} \linelabel{#3end}  }

\newcommand\textreference[3]{\linelabel{#2} #1 \comm{C\ref{#3}} \linelabel{#2end}  }

\newcommand{\comm}[1]{\fcolorbox{black}{yellow}{\bfseries\sffamily\scriptsize#1}}

\renewcommand\addcitation[3]{\cite{#1}}
\renewcommand\deletecitation[3]{}
\renewcommand\delete[3]{}
\renewcommand\rep[4]{#1}
\renewcommand\change[3]{#1}
\renewcommand\changeRead[3]{#1}
\renewcommand\textreference[3]{}
 
\definecolor{blue}{RGB}{0,0,0}

\section{Introduction}
\label{sec:intro}

\emph{Search-based software testing} (SBST) automatically generates test cases searching for violations of software requirements~\cite{harman2001search}. 
SBST is widely used in many domains, such as real-time, concurrent, distributed, embedded, safety-critical, and cyber-physical systems (CPS)~\cite{5210118,stocco2023model,liu2019effective,9869302,jahangirova2021quality,arrieta2016test}. 
SBST is also applied within \simulink~\cite{Simulink}, a widely used tool to design software~\cite{elberzhager2013analysis,DBLP:journals/sosym/LiebelMTLH18} used in many industries, such as medical, avionics, and automotive~\cite{dajsuren2013simulink,boll2021characteristics}.
Developing SBST techniques for Simulink models is a significant and widely recognized software engineering problem~\cite{Aristeo,menghi2019generating,kapinski2016simulation} of interest to both academia and industry. 
However, the existing SBST techniques for Simulink models (e.g.,~\cite{Aristeo,menghi2019generating}) mostly consider requirements expressed in a temporal logical format hampering a wider adoption of these techniques in practice.
Indeed, temporal logical specifications are difficult to write even for expert logicians (e.g.,~\cite{Greeman_2025,dokhanchi2015metric,autili2007graphical,rozier2016specification,maoz2018software,menghi2019specification,dwyer1999patterns}).

Recent work proposed an SBST technique that considers Test Sequence and Test Assessment blocks from \SLTest~\cite{SimulinkTest} to generate failure-revealing test cases~\cite{Hecate2024}. 
The solution was evaluated considering 18 Simulink models from different domains and industries. 
It also assessed the effectiveness of this technique by developing a cruise controller model powered by an industrial simulator~\cite{10.1145/3611643.3613894}.
Despite its advantages, this technique requires engineers to translate requirements into Simulink Test Assessment blocks.
Proposing solutions that \emph{directly consider software requirements} to generate failure-revealing test cases would make SBST more accessible and save development time. 
It will also prevent engineers from designing erroneous Test Assessment blocks for the software requirements.

\change{Requirements Tables (RTs)~\cite{RequirementsTable,menghi2024completeness} from the \ReqTbx~\cite{RequirementsToolbox} are a powerful tool to specify requirements as part of the Simulink model. 
RTs are a specific form of \emph{tabular requirements specification} (e.g.~\cite{heninger1980specifying,faulk1990state,meyers1983software}), i.e., requirements specifications in a table format. 
Tabular requirements have been extensively studied by the academic community (e.g.,~\cite{heitmeyer1996automated,DBLP:conf/ftscs/PangWLW13,singh2017use,peters2007ide,eles2011tabular,menghi2024completeness}). They were used to specify the software requirements of industrial systems (e.g.,~\cite{crow1998formalizing,10.1145/2038642.2038676}) including the space shuttle~\cite{crow1998formalizing}, the shutdown system of the Darlington Canadian Nuclear Generation Station~\cite{10.1145/2038642.2038676}, the US Navy's A-7 aircraft~\cite{heninger1980specifying}, and an Ericsson telecom software system~\cite{quinn2006specification}.}{tech_difficulties_stateflowAnswer}{tech_difficulties_stateflow}

This work proposes Requirements Table Driven SBST, a framework that generates test cases from RT.
\change{Our framework relies on a quantitative semantics of RT, which is one of the contributions of this work.}{QuantSemanticsContribution}{small_contribution_1}

Our solution compiles RT into Simulink Stateflow~\cite{SimulinkStateflow} blocks, providing a fitness measure to guide the search process.
Stateflow is Mathworks' framework for specifying state machines.
We implemented our solution within HECATE~\cite{Hecate2024}, a testing tool that supports Test Sequence and Test Assessment blocks from \SLTest.

We assessed the effectiveness and efficiency of our solution using three different Simulink models: the Automatic Transmission Controller~\cite{ARCH14} from the  Applied Verification for Continuous and Hybrid Systems
(ARCH)~\cite{Arch2024}, 
the model of a Cruise Controller (CC) for an industrial simulator~\cite{10.1145/3611643.3613894}, and a model from the Simulink documentation~\cite{RequirementsTable}.
For each model, we manually define variations of the models and their requirement set specified in an RT.
We then pair each model version with each version of the RT for the same model and apply our solution.
This process results in testing \modelrtcombinations different model-RT combinations over the three starting models.
Our SBST framework could effectively find failure-revealing test cases for \rqonetotal of the model-RT combinations.
Remarkably, it identified a failure-revealing test case for three model-RT combinations from the CC example that the tool used in the original publication failed to identify.
We justified why our SBST could not reveal failures for $15\%$ of our model-RT combinations.
Also, our SBST framework was efficient in finding failure-revealing test cases: on average (across all the model-RT combinations for each Simulink model), our SBST framework required respectively $73.4$, $6.6$, and $24.1$ iterations,  achieving performance comparable to other SBST frameworks not based on RTs.
We also informally compare our SBST and other solutions supporting other specification languages, and justify why an experimental comparison is not relevant for this work.
Finally, our solution supports different search algorithms. While the implementation of these algorithms is not part of our contribution, \changeRead{we do provide a comparison between the effectiveness of Uniform Random and Simulated Annealing search procedures.}{search_methodAnswerZ}{search_method}

\change{To summarize, our contributions are as follows: 
\begin{itemize}
\item[\textbf{C1}:] We propose Stateflow Charts as the target tool to encode our fitness metric --- after assessing different alternatives (e.g., logic-based and programming languages). 
     Stateflow Charts enable the development of a translation (a)~fully integrated within the \simulink environment, (b)~with a clean and simple formalization that simplifies its presentation, description, and the verification of its soundness, and (c)~producing monitors that assess the requirements as the simulation is running (\Cref{sec:hecate}).
    \item[\textbf{C2}:]  We define a quantitative semantics for the satisfaction of RT and their requirements, and we propose its use as a fitness metric to guide the search-based process. 
    The quantitative semantics extends the two-valued semantics of the RT~\cite{menghi2024completeness} with a quantitative measure for the satisfaction and violation of the requirements (\Cref{sec:Semantics}). 
    \item[\textbf{C3}:] We define a translation to convert RTs into fitness functions represented as Stateflow Charts (\Cref{sec:transl}). 
    Our translation overcomes several challenges: (a)~The encoding of the multiple requirements of the RTs, which must be checked simultaneously (\Cref{sec:algorithm}),
    (b)~the encoding of the \prev (\Cref{sec:dur}) and \duration (\Cref{sec:prev}) operators that can be used to express the preconditions and the postconditions of the requirements, which are not available within Stateflow Charts, and (c)~the definition of a proof that our translation leads to a fitness function that satisfies a set of well-formedness properties (\Cref{sec:soundness}).
    \item[\textbf{C4}:] We implemented our solution as a plugin for HECATE~\cite{Hecate2024} since it is the only testing approach supporting the Test Sequences for test case generation. We decided to consider Test Sequence blocks~\cite{TestSequenceBasics} since they are widely used within the  \simulink environment to specify test cases~\cite{formica2022search}.
    \item[\textbf{C5}:] We defined a benchmark of RT for SBST: since RTs were created only recently (released at the beginning of 2022 \cite{RequirementsTable}), a benchmark dedicated to SBST does not exist.
    Our benchmark contains \modelrtcombinations model-RT combinations from three practical examples (\Cref{sec:benchmark}).
     We provided a complete replication package including our tool, the benchmark RT, and scripts for replicating our experiments~\cite{ReplicationPackage}.
    \item[\textbf{C6}:] We empirically evaluated the effectiveness (\Cref{sec:effectiveness}), i.e., the capability of generating failure-revealing test cases, and efficiency (\Cref{sec:efficiency}), i.e., the number of iterations required to generate the failure-revealing test cases, of our solution.
    Our evaluation required us to study and understand the three practical examples.
    For each practical example, we critically analyzed the results obtained for the corresponding model-RT combinations.
    For one of them, we verified the obtained results via hardware-in-the-loop testing, which required a considerable effort to set up the hardware testing platform.
    \item[\textbf{C7}:] We discuss the relevance of our contribution and its industrial applicability (\Cref{sec:relevance}), present threats to validity (\Cref{sec:threats}), and compare of our solution with existing ones from the literature (\Cref{sec:comparison}).
\end{itemize}
}{small_contribution_1AnswerCont}{small_contribution_1}\change{}{small_contribution_2AnswerCont}{small_contribution_2}
\change{}{}{small_contribution_3}

This paper is organized as follows.
\Cref{sec:background} presents background notions by introducing RT, Simulink, and Simulink Stateflow.
\Cref{sec:hecate} presents our SBST framework.
\change{\Cref{sec:Semantics} defines a quantitative semantics for RT}{paperStructure}{algorithm_completeness}
\Cref{sec:transl} describes the translation of RT into a fitness function. 
\Cref{sec:eval} evaluates our solution.
\Cref{sec:related} presents the related work.
\Cref{sec:conclusion} presents our conclusions. 
\begin{figure*}[t!]
    \centering
    \begin{subfigure}[b]{0.7\textwidth}
        \centering
        \includegraphics[width=1\linewidth]{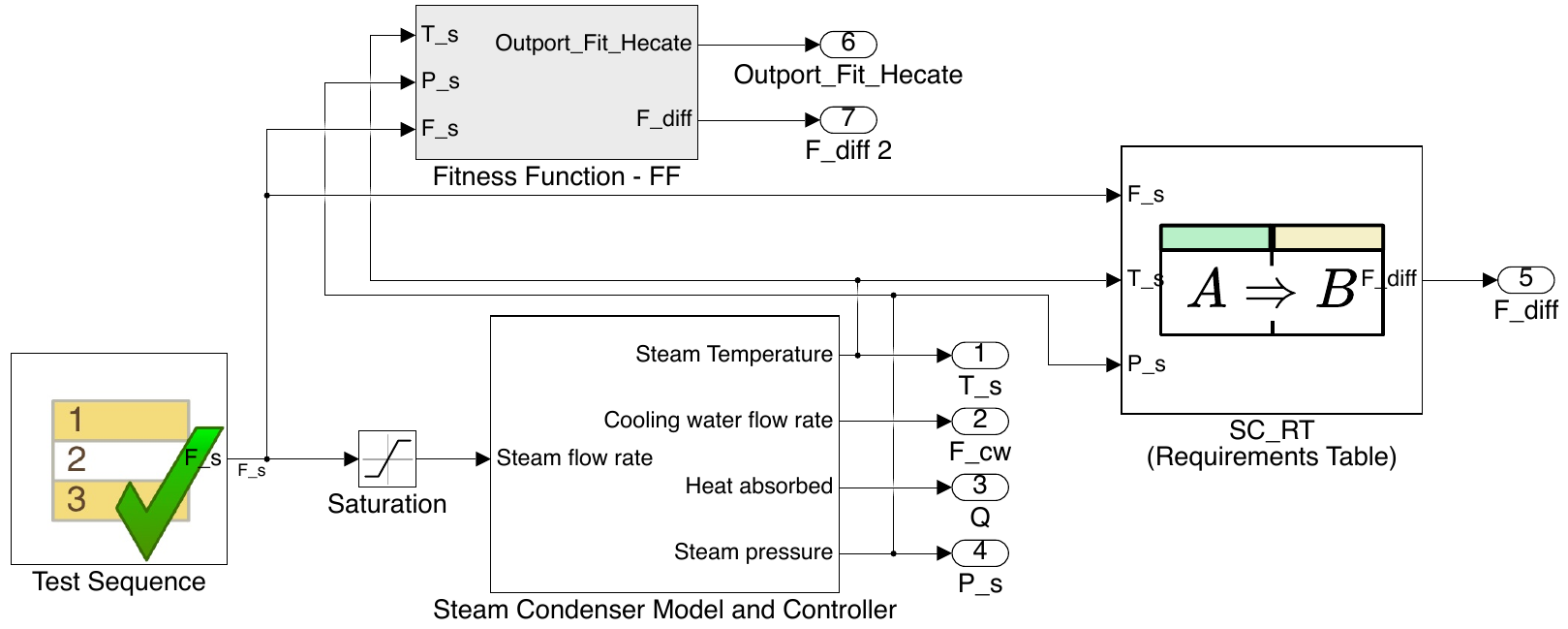}
        \caption{The Simulink model of the Steam Condenser (SC) and its controller.}
        \label{fig:exampleModel}
    \end{subfigure}\\
    \begin{subfigure}[b]{0.7\textwidth}
        \centering
        \includegraphics[width=1\linewidth]{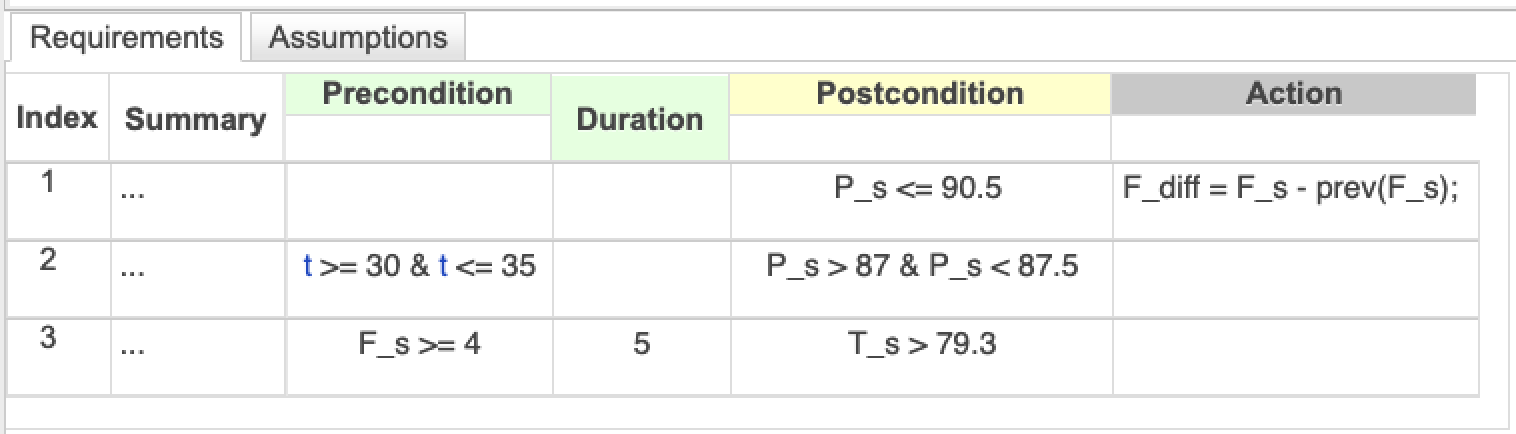}
        \caption{The Requirements Table.}
        \label{fig:exampleRT}
    \end{subfigure}\\
    \begin{subfigure}[b]{0.64\textwidth}
        \centering
        \includegraphics[width=1\linewidth]{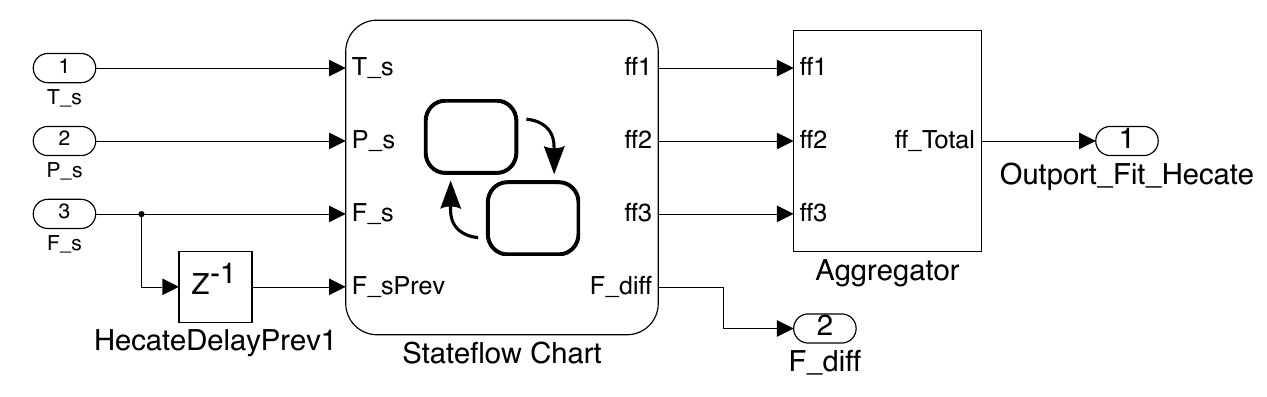}
        \caption{Fitness Function (FF) subsystem}
        \label{fig:FitCalculator}
        \end{subfigure}
    \begin{subfigure}[b]{0.34\textwidth}
        \centering
        \includegraphics[width=0.9\linewidth]{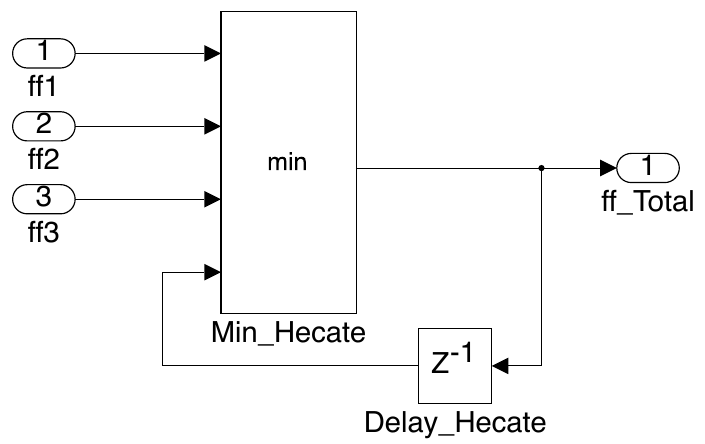}
        \caption{Fitness function Aggregator.}
        \label{fig:FitConverter}    
        \end{subfigure}\\
    \begin{subfigure}[b]{0.7\textwidth}
        \centering
        \includegraphics[width=1.0\linewidth]{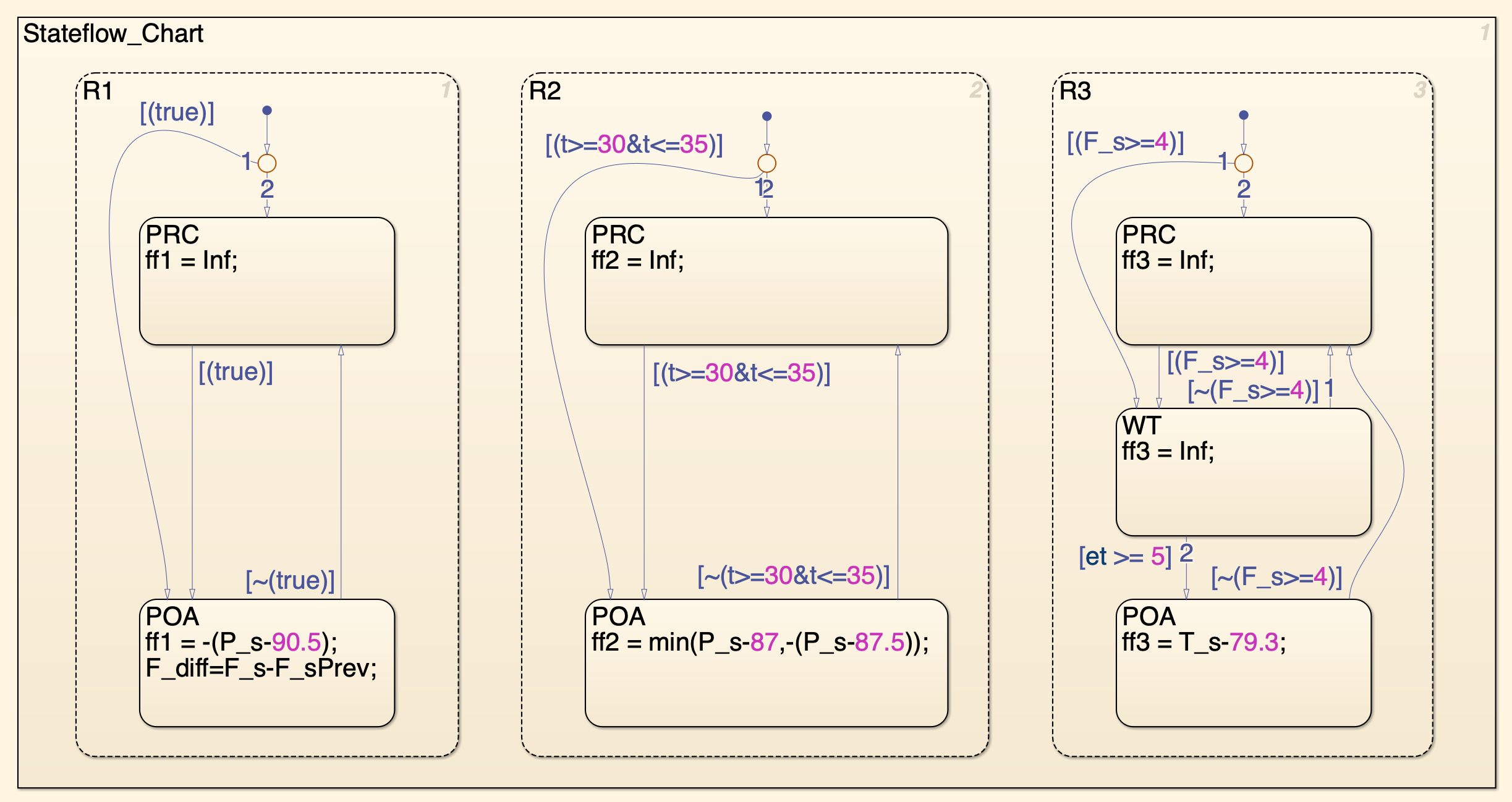}
        \caption{Example of a Stateflow model.}
        \label{fig:exampleSF}
    \end{subfigure}
    \caption{The Steam Condenser running example.}
    \label{fig:example}
\end{figure*}

\section{Simulink, Requirements Tables,  and Simulink Stateflow}
\label{sec:background}
We introduce  
Simulink (\Cref{sec:Simulink}), Requirements Tables (\Cref{sec:reqtable}), 
 and Stateflow (\Cref{sec:Stateflow}).

\subsection{Simulink}
\label{sec:Simulink}
Simulink~\cite{Simulink} is a modeling language for software design based on signals.
For example, \Cref{fig:exampleModel} presents the steam condenser (SC)~\cite{YaghoubiHSCC} example, which is used as a benchmark model for the ARCH-COMP annual  competition~\cite{Arch2024}.
The model has one input, the steam flow rate ($\mathit{F\_s}$), and four outputs: the steam temperature ($\mathit{T\_s}$), the cooling water flow rate ($\mathit{F\_cw}$), the heat absorbed ($Q$) and the steam pressure ($\mathit{P\_s}$).
The SC software controller regulates the cooling water flow rate to maintain steam temperature and pressure within acceptable levels.
Simulink models consist of components, represented by Simulink blocks, linked by connections.
The SC model has $173$ blocks.
For example, the steam flow rate ($\mathit{F\_s}$) input is linked to the \emph{Saturation} component, and the output of the \emph{Saturation} component is linked to the \emph{Steam Condenser Model and Controller}.
A component can represent a subsystem and be refined using other components. 
For example, the \emph{Steam Condenser Model and Controller} contains the model of the steam condenser and its controller.
The Simulink model also embeds the \emph{Test Sequence} and the \emph{Requirements Table} blocks.
The \emph{Test Sequence} block is used for test case specification and is detailed below.
The \emph{Requirements Table} block (SC\_RT) is used for requirements specification and is described in \Cref{sec:reqtable}.
The model also uses \emph{Inport} and \emph{Outport} blocks to transmit signals into and from subsystems.
For example, the outport \emph{Outport\_Fit\_Hecate} in \Cref{fig:FitCalculator} is used to get the signal \emph{ff\_Total} outside of the \emph{Fitness Function} subcomponent to the root level of the model.
Outports on the root level of the model send the signals to the \matlab workspace. 
Our Simulink model also embeds the fitness function (FF) in an appropriate Simulink component.
As we will describe in \Cref{sec:hecate}, this Simulink component is generated by our SBST framework and will guide the search-based testing framework.
The \emph{Fitness Function} component is detailed in \Cref{fig:FitCalculator} and consists of two subcomponents (a \emph{Stateflow Chart} and an \emph{Aggregator}).
The \emph{Aggregator} component is detailed in \Cref{fig:FitConverter}.
These two components will be described in \Cref{sec:transl}.

\begin{figure*}[t]
    \centering
    \hfill
    \begin{subfigure}{0.47\textwidth}
        \includegraphics[width=0.8\textwidth]{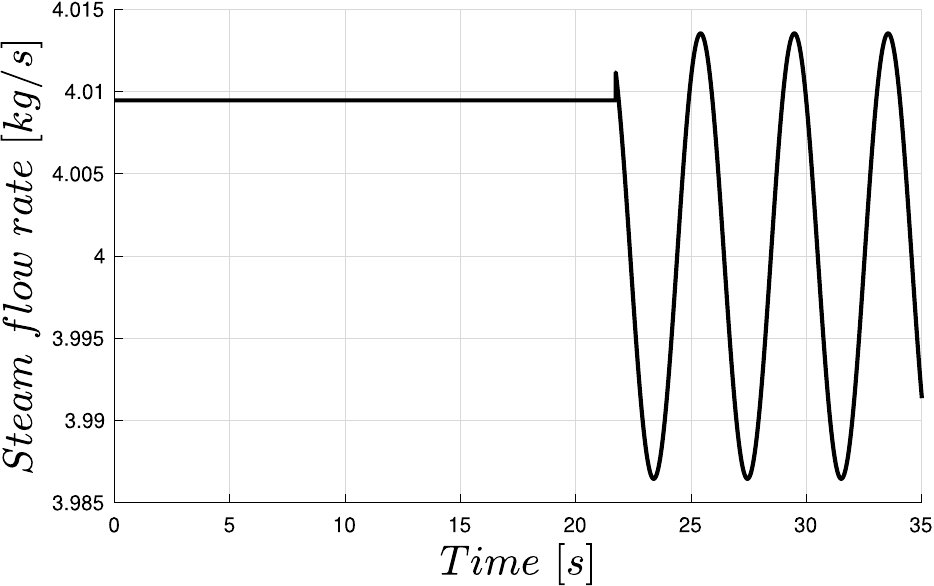}
        \caption{Inputs}
        \label{fig:in}
    \end{subfigure}
    \hfill
    \begin{subfigure}{0.47\textwidth}
        \includegraphics[width=0.8\textwidth]{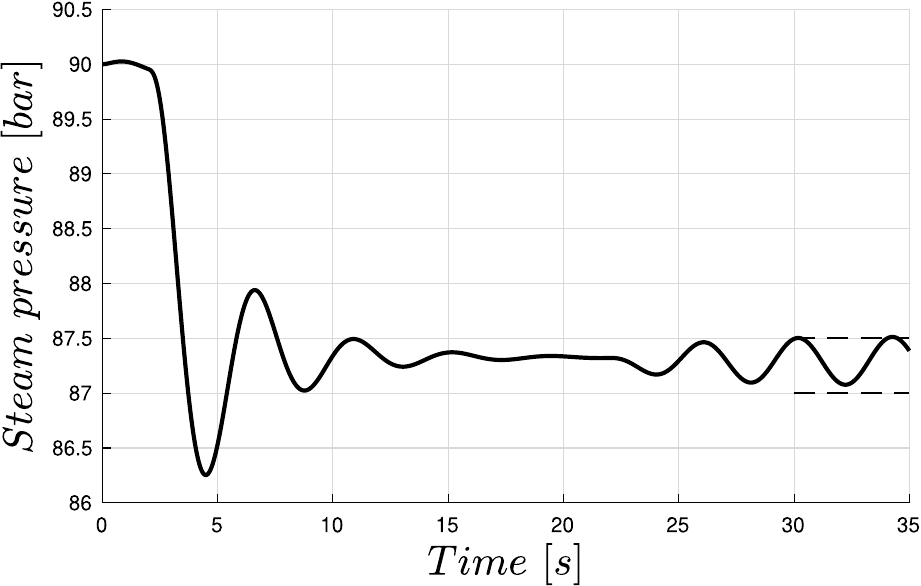}
        \caption{Outputs}
        \label{fig:out}
    \end{subfigure}
    \hfill
    \caption{An example of Input/Output signals of the Steam Condenser (SC) Simulink model.}
    \label{fig:signals}
\end{figure*}

The Simulink simulator considers a set of input signals (one for each input of the model) and it produces a set of output signals (one for each output of the model) that depend on the behavior of the software components.
In this work, we assume Simulink Test Sequences are used to generate input signals.
The Test Sequence blocks are supported by \SLTest. 
They are a standard mechanism for specifying test input in Simulink.
They provide a graphical notation to specify a test input containing one input signal for each model input. 
However, although they are used for test input specification in this work, their details are not relevant to our work: The interested reader can refer to online resources for additional information~\cite{Hecate2024}.
For example, \Cref{fig:in} shows a test input for the SC example generated by a Test Sequence showing how the steam flow rate changes over the time domain $[0, 35]$s.
\Cref{fig:out} plots one of the output signals produced by the Simulink simulator showing how the steam pressure changes over time for the considered test input.

\begin{figure}
    \footnotesize
\begin{align}
&    \term  &&::= && \constant \mid \variable \mid \prev(\variable) \mid \term_1 \odot \term_2  &\nonumber\\
&     \logicalexpression &&::=&& \term_1\oplus \term_2 \mid \neg \logicalexpression \mid \duration(\logicalexpression) \geq c \mid  \logicalexpression_1 \oslash \logicalexpression_2 & \nonumber \\
&      \preconditions && :: =&& \logicalexpression &\nonumber \\
&    \postconditions && :: =&& \logicalexpression& \nonumber \\
&     \requirement && ::= && \preconditions  \Rightarrow \postconditions& \nonumber \\
&     \requirementTable && ::= && \requirement \mid \requirement, \requirementTable & \nonumber 
    \end{align}
    $c\in \real$, $\variable \in \variables$,
    $\odot \in \{+, -, * , /\}$,
$\oplus \in \{>, <, \leq, \geq, =,\not=\}$, $\oslash \in \{\&, |, \Rightarrow\}$.

     \caption{Requirements Tables: Syntax}
    \label{fig:syntax}
\end{figure}

\subsection{Requirements Tables}
\label{sec:reqtable}
Engineers can use Requirements Tables (RTs)~\cite{ReqTableDocumentation,RequirementsTable,menghi2024completeness} to specify the requirements of Simulink models.
RTs are directly embedded into the model.
For example, the RT for the SC is embedded within the Simulink model in \Cref{fig:exampleModel}.
The input of the RT from \Cref{fig:exampleModel} takes the steam flow rate ($\mathit{F\_s}$), the steam temperature ($\mathit{T\_s}$), and the steam pressure ($\mathit{P\_s}$) as inputs; its output is the difference in the flow rate ($\mathit{F\_diff}$). 
For example, the RT for the SC from \Cref{fig:exampleModel} is detailed in \Cref{fig:exampleRT}.

\change{\Cref{fig:syntax} presents the formal syntax of RT~\cite{menghi2024completeness}.}{RT_syntax}{small_contribution_1} In practice, RTs are enriched with graphical coloring, an index to identify the different requirements, a summary, the graphical separation of pre- and post-conditions, and the duration column. 

The RT from \Cref{fig:exampleRT} consists of three requirements associated with the indexes (identifiers) 1, 2, and 3: Each row of the table details one requirement for the SC.
\change{This is represented in the grammar from \Cref{fig:syntax} by defining an RT (\requirementTable) as a requirement~(\requirement) or a requirement~(\requirement) concatenated with a Requirements Table~(\requirementTable).}{RequirementRows}{small_contribution_1}
Requirements 1 and 3 are two requirements introduced for illustration purposes, while Requirement 2 represents the only requirement considered in the ARCH-COMP competition.

The preconditions are conditions on the inputs of the model and the variable $t$ used to encode the simulation time.
For example, the precondition for the requirements~2 and~3 from \Cref{fig:exampleRT} respectively assesses if the simulation time (t) is between $30$ and $35$ seconds (the boundary values $30$ and $35$ are included) and if the steam flow rate ($\mathit{F\_s}$) is greater than or equal to four.
When the precondition of a requirement is satisfied, the postcondition should be satisfied. 
For example, when the precondition of the second requirement is satisfied,  the value of the steam pressure should be between $87$ and $87.5$ (the boundary values $87$ and $87.5$ are excluded).
Preconditions \preconditions and postconditions \postconditions are logical expressions.
The duration column~\cite{ReqTableDocumentationDuration} of the RT enables engineers to specify a duration for which the precondition should be satisfied before the postcondition is evaluated. 
For example, the RT from \Cref{fig:exampleRT} specifies that the precondition of Requirement~3 should be satisfied for at least five seconds before the system has to satisfy the postcondition on its output signals.
\change{The grammar from \Cref{fig:syntax} defines a requirement \requirement as an expression in the form $\preconditions[\durationcolumn]  \Rightarrow \postconditions$ that relates a precondition \preconditions with a postcondition \postconditions.
The optional duration column is encoded by the term \durationcolumn. 
The duration column expresses the time (in seconds) during which the precondition must be true before evaluating the rest of the requirement.}{DurationColumn}{small_contribution_1}

The pre- and post-conditions consist of Boolean operators among expressions.
For example, the postcondition $\mathit{P\_s} > 87~\&~\mathit{P\_s} < 87.5$ is obtained by applying the Boolean operator ``\&'' to the expressions $\mathit{P\_s} > 87$ and $\mathit{P\_s} < 87.5$.
\change{Formally, the grammar from \Cref{fig:syntax} defines a logical expression \logicalexpression as either 
  the combination $\term_1\oplus \term_2$ of two terms $\term_1$ and $\term_2$ with a relational operator $\oplus \in \{>, <,\leq,\geq, =, \not=\}$, 
  the negation $\neg \logicalexpression$ of a logical expression $\logicalexpression$, 
  the duration term $\duration(\logicalexpression) \geq \constant$ indicating that a logical expression $\logicalexpression$ holds for at least $\constant$ seconds, 
  or the combination  $\logicalexpression_1 \oslash \logicalexpression_2$ of two logical expressions obtained by considering the Boolean operators $\oslash \in \{\&, |, \Rightarrow\}$.}{RelationalOperators}{small_contribution_1}
For example, the atom $\mathit{P\_s} > 87$ of the postcondition $\mathit{P\_s} > 87~\&~\mathit{P\_s} < 87.5$ relates two expressions respectively representing the variable $\mathit{P\_s}$ and the constant $87$.

Even though this is not the case for our running example, each arithmetic expression can use arithmetic operators (i.e., $\term_1 \odot \term_2$ with $\odot \in \{+, -, *, /\}$).
For example, the atom $2+\mathit{P\_s} > 87$  could be used as an atom for the postcondition.
RT also enables engineers to use the \prev (a.k.a 
\emph{getPrevious}) operator~\cite{Prev}.
This operator returns the value of the data at the previous time step.
For example, the value of $\prev(\mathit{F\_s})$ within the action column of the RT represents the value of $\mathit{F\_s}$ in its previous simulation time.
Note that, starting from Simulink R2023a, the initial values of inputs and outputs must be specified to be used as the previous values for the beginning of the simulation to avoid undefined behavior.
\change{Formally, the grammar from \Cref{fig:syntax}  defines a \emph{term} \term as either a real constant (i.e., \constant), an input or output variable (i.e., \variable), the term $\prev(\inputVariable)$ indicating the previous value of the input variable \inputVariable, 
 or a combination $\term_1 \odot \term_2$ of two terms $\term_1$, $\term_2$ with an arithmetic operator $\odot \in \{+,-,*,/\}$.}{TermDefinition}{small_contribution_1}

RTs can specify actions in their action columns. 
Actions are operations that should be executed when their precondition is satisfied.
For example, the action $\mathit{F\_diff}=\mathit{F\_s}-\mathit{prev}(\mathit{F\_s})$ sets the value of the output signal $\mathit{F\_diff}$ to $\mathit{F\_s}-\mathit{prev}(\mathit{F\_s})$ when the precondition of Requirement~1 is satisfied. 
Since the precondition of  Requirement~1 is always satisfied, the action sets the value of the variable $\mathit{F\_diff}$ at every simulation time. 
Otherwise, there would exist a simulation time for which no action specifies a value of $\mathit{F\_diff}$, and the simulation would stop with an error.
\change{The grammar from \Cref{fig:syntax}  does not account for actions since they can be considered as a part of the post-condition, i.e., the post-condition of the RT can include the conjunction of the content of the Postcondition and the Action columns from \Cref{fig:exampleRT}.}{ActionColumn}{small_contribution_1}

Note that several preconditions can be satisfied simultaneously, forcing the satisfaction of the postcondition of the corresponding requirements.
For example, for the RT from \Cref{fig:exampleRT} the postcondition of Requirement~1 should always be satisfied since the requirement comes with no precondition.
The postconditions of Requirements~2 and 3 should be satisfied when the condition of the corresponding preconditions holds.

\rep{The interested reader can refer to a recent work~\cite{menghi2024completeness} for additional details on the syntax of RTs.}{A recent work~\cite{menghi2024completeness} defines the formal semantics of RT, which provides additional details for the interested reader.}{RephrasingCitation}{small_contribution_1}

\subsection{Stateflow}
\label{sec:Stateflow}
Stateflow~\cite{Stateflow} is a graphical language used to model state transition diagrams, such as state machines.
\Cref{fig:exampleSF} presents an example of a state machine.
This section does not exhaustively describe all the constructs that can be used to represent state machines but introduces only the ones that will be used by our SBST solution.
We will first describe the state machine R2 in \Cref{fig:exampleSF}. 
Then, we will describe its relationship with the other state machines (i.e., R1 and R3).

A state machine consists of a (finite) set of states and transitions that specify how the system moves from one state to another in response to events. 
For example, the box labeled R2 in \Cref{fig:exampleSF} represents a state machine consisting of two states  PRC and POA.

Links ended by arrows represent transitions that change the state of the Stateflow.
The direction of the arrow indicates the destination of the transition.
For example, if the system is in state PRC, the satisfaction of the condition $t>=30~\&~t<=35$ triggers an event that moves the system from state PRC to POA.
Note that the conditions that label transitions can also use the value of the reserved variable named $et$ that represents the elapsed time, i.e., the length of time that has elapsed since the associated state was entered.
For example, the condition $\mathit{et}>=5$ that labels the transition from state WT to POA of the state machine R3 is triggered
if the elapsed time since the WT was entered is greater than or equal to $5$.

States are labeled with actions that change the values of variables. 
For example, the action $\mathit{ff2}=\mathit{Inf}$ sets the values of the variable $\mathit{ff2}$ to the value $\mathit{Inf}$ (infinite).
Note that when a Stateflow fires a transition reaching its destination state, it also executes its action.

The blue-colored circle is the entry point of the state machine, i.e., when the system starts it enters the blue-colored circle.
For example, the blue-colored circle connected to the red empty circle from \Cref{fig:exampleSF} indicates the starting point of the finite state machine.

The red empty circle is a connective junction representing decision points, i.e., \texttt{if/else} conditions.
For example, if the condition $t>=30~\&~t<=35$ in R2 from \Cref{fig:exampleSF} is satisfied, the system moves to the POA state, otherwise it moves to the PRC state.
Therefore, depending on the initial value assumed by the variable $t$, the system will start from the PRC or POA state.

States can be hierarchically decomposed via exclusive and parallel states.
This work considers parallel states representing independent operations executed in parallel.
For example, the RT state of the Stateflow from \Cref{fig:exampleSF} is decomposed into three parallel states: R1, R2, and R3 which are executed in parallel.

\Cref{sec:hecate} presents an overview of our SBST framework.
\Cref{sec:Semantics} presents the quantitative semantics of RT that is used to drive the search process.
\Cref{sec:transl} will describe how RTs are converted into Stateflow models (that reflect their quantitative semantics) and compute the fitness function used by our SBST framework.

 \section{Requirements Table Driven SBST}
\label{sec:hecate}

\tikzstyle{output} = [coordinate]
\begin{figure}[t]
\centering
\begin{tikzpicture}[auto,
 block/.style ={rectangle, draw=black, thick, fill=white!20, text width=5em,align=center, rounded corners},
 block1/.style ={rectangle, draw=blue, thick, fill=blue!20, text width=5em,align=center, rounded corners, minimum height=2em},
 line/.style ={draw, thick, -latex',shorten >=2pt},
 cloud/.style ={draw=red, thick, ellipse,fill=red!20,
 minimum height=1em}]

\node [block,node distance=2.2cm,text width=2.7cm] (FFG) {\phase{1} \footnotesize Fitness-function\\Generator};
\node [output, left of=FFG,node distance=2.5cm] (RT) {};
\node [block,right of=FFG,node distance=3.5cm,text width=2.0cm] (SE) {\phase{2} \footnotesize Search\\ Engine};
\node [output, above of=SE,node distance=1cm] (PITB) {};
\node [output, right of=SE,node distance=2.5cm] (OUT) {};

\draw[-stealth] (RT.east) -- (FFG.west)    node[pos=0.5, above]{RT};
\draw[-stealth] (FFG.east) -- (SE.west)    node[pos=0.5, above]{FF};
\draw[-stealth] (PITB.south) -- (SE.north)    node[pos=0.5, right]{PI,TB};
\draw[-stealth] (SE.east) -- (OUT.west)    node[pos=0.5, above]{TC/NFF};
\end{tikzpicture}
\caption{Overview of our SBST framework.}
\label{fig:contribution}
\end{figure}

\Cref{fig:contribution} presents a high-level view of our SBST framework.
The inputs to the framework are: a parameterized input (PI), a Requirements Table (RT), and a time budget (TB). 
The output is either a failure-revealing test case (\failingtestsequence) or a value (\nff) indicating that the SBST framework could not detect a failure-revealing test case within the time budget.
Our SBST framework comprises the \emph{Fitness-Function Generator} (\phase{1}) and the  \emph{Search Engine} (\phase{2}) phases.

The \emph{Fitness-Function Generator} phase (\phase{1}) compiles the 
Requirements Table to a fitness function (FF) that guides the search-based exploration.
For example, \Cref{fig:FitCalculator} shows the FF for the SC running example. 
The fitness function generated from the RT should satisfy the following well-formedness properties: 
\begin{enumerate}
    \item a negative fitness value indicates that the RT is violated; 
    \item a positive fitness value indicates that the RT is satisfied;
    \item the higher the positive fitness value is, the farther the RT is from violation; and 
    \item the lower the negative value is, the farther the RT from satisfaction.
\end{enumerate}
These properties can be interpreted differently.
We report two interpretations for Requirements~1, 2, and 3 of the SC running example (many other interpretations exist).
Consider two scenarios.
In the first scenario, Requirements~1 and 2 are far from being violated (e.g., $\mathit{P\_s}=87.25$), but Requirement 3  is very close to being violated (e.g., $\mathit{T\_s}=79.300001$).
In the second scenario, all the Requirements~1, 2, and 3 are close to the violation (e.g., $\mathit{P\_s}=87.45$ and $\mathit{T\_s}=79.35$).
One interpretation may consider the first scenario as the one in which the system is closer to violating its requirements since $79.300001$ is just about reaching the boundary value $79.3$: None of the requirements for the second scenario is as close as Requirement~$3$ from violating its requirement. 
Another interpretation may consider the second scenario closer to violation since the values $\mathit{P\_s}=87.45$ and $\mathit{T\_s}=79.35$ are closer to their thresholds $\mathit{P\_s}=87.5$ and $\mathit{T\_s}=79.3$ on average.
Both these interpretations are valid (as are many others).
In this work, we select the first interpretation. 
\begin{definition}
\label[definition]{def:multiple}
Let \requirementTable be an RT made of $n$ different requirements  $\requirement_1, \requirement_2, \ldots, \requirement_n$ and $\interpretation{\requirement_1},$ $\interpretation{\requirement_2}, \ldots,$ $\interpretation{\requirement_n}$ be their satisfaction degrees. Then, the satisfaction degree associated with the RT \requirementTable is the minimum across the satisfaction degrees of its requirements, i.e., $\interpretation{\requirementTable}=min(\interpretation{\requirement_1},$ $\interpretation{\requirement_2}, \ldots,$ $\interpretation{\requirement_n})$.
\end{definition}
According to this definition, we can reformulate the properties of a fitness function as follows:
\begin{definition}
\label[definition]{def:ftprop}
A fitness function of an RT is a function from a set of input signals to a fitness value that satisfies the following \emph{well-formedness} properties:    
\begin{enumerate}
    \item a negative fitness value indicates that at least one requirement from the RT is violated; 
    \item a positive fitness value indicates all the requirements of the RT are satisfied;
    \item the higher the positive fitness value is, the farther the system is from violating the requirement that is closer to violation; and 
    \item the lower the negative value is, the farther the system is from satisfying the requirement that is farther from satisfaction.
\end{enumerate}
\end{definition}
The goal of the \emph{Fitness-Function Generator} phase (\phase{1}) is to compile the Requirements Table into a fitness function (FF) that provides these properties
\rep{. To define this component, we first define a quantitative semantics for RT (Section \ref{sec:Semantics}) which is then as a metric to guide the search process, and therefore, considered for developing the translation from RT to FF (Section \ref{sec:transl}).}
{and is described in Section \ref{sec:transl}.}
{introQuantSemantics}
{small_contribution_1}

The \emph{Search Engine} (\phase{2}) phase iteratively generates a test case by assigning values to the parameters of the parameterized inputs (PI).
Then, it executes the model for this test input and computes the fitness value using the fitness function.
If the test case leads to a negative fitness value, it is a failure-revealing test case and it is returned as output.
Otherwise, a new test case is generated.
The test cases previously generated and their corresponding fitness values computed by the fitness function drive the generation of the values assigned to the parameters of the PI.
If no test case with a negative fitness value is found within the considered time budget (TB), the NFF value is returned, indicating that the search engine could not find test cases that violate the requirements within the considered budget. 
\changeRead{The search engine is not a contribution provided by this work. 
The academic and industrial communities developed many search engines.
We reuse off-the-shelf engines to implement our framework (see \Cref{sec:eval}): Reusing existing components with proven efficiency is a recommended software engineering practice.}{search_methodAnswer}{search_method}

We considered Test Sequence blocks~\cite{TestSequenceBasics} for the specification of test inputs, and their parameterized version (Parameterized Test Sequences~\cite{Hecate2024}) for the specification of the parameterized input.
Therefore, we developed our solution as a plugin for HECATE~\cite{Hecate2024}, a testing approach supporting the Test Sequence and Test Assessment blocks from \SLTest.
This decision enables us to reuse the \emph{Search Engine} component from HECATE, which in turn relies on S-Taliro~\cite{STaliro}, a widely known testing tool for Simulink models.
We integrated the \emph{Fitness-Function Generator} component into HECATE as described in \Cref{sec:transl}. 

 \section{Quantitative Semantics of Requirements Tables}
\label{sec:Semantics}
\change{
First, we define traces (\Cref{sec:traces}). 
Traces are used to define the quantitative (robustness) semantics of Requirements Tables (\Cref{sec:semantics}) that drive the definition of our fitness metric.}{introSemantics}{algorithm_completeness}

\subsection{Traces}
\label{sec:traces}
\change{
A trace $\pi$ associates a value for every input $\inputVariable \in \inputs$ or output $\outputVariable \in \outputs$ variable and position $\indexvariable \in \nonnegativenatural$.
For example, \Cref{fig:traceExample} presents two examples of (fragments of) traces for the instance from \Cref{fig:signals}.}{traceSectionStart}{algorithm_completeness}

\begin{figure*}[t]
\centering
 \begin{subfigure}[b]{\columnwidth}
    \scalebox{0.9}{
       \begin{tikzpicture}
	\pgfmathsetmacro{\ilocationangularrate}{0.25}
	\pgfmathsetmacro{\ilocationmode}{-0.25}
	\pgfmathsetmacro{\ilocationtimestamp}{-0.75}
	\pgfmathsetmacro{\ilocationindex}{-1.2}
	\draw[dashed] (-2,0.5) -- (7.2,0.5);
		\draw node at (-1,\ilocationangularrate) {\footnotesize $F\_s~[kg/s]$};
	\draw node at (0.5,\ilocationangularrate) {\small $4.007$};
	\draw node at (2,\ilocationangularrate) {\small $4.005$};
	\draw node at (3.5,\ilocationangularrate) {\small $4.003$};
	\draw node at (5,\ilocationangularrate) {\small $4.001$};
	\draw node at (6.5,\ilocationangularrate) {\small $3.999$};
	\draw[dashed] (-2,-0) -- (7.2,-0);
		\draw node at (-1,\ilocationmode) {\footnotesize $P\_s~[bar]$};
	\draw node at (0.5,\ilocationmode) {\small $87.321$};
	\draw node at (2,\ilocationmode) {\small $87.321$};
	\draw node at (3.5,\ilocationmode) {\small $87.321$};
	\draw node at (5,\ilocationmode) {\small $87.320$};
	\draw node at (6.5,\ilocationmode) {\small $87.317$};
		\draw[dashed] (-2,-0.5) -- (7.2,-0.5);
		\draw node at (-1,\ilocationtimestamp) {\small $\timevariableEncoding~[s]$};
	\draw node at (0.5,\ilocationtimestamp) {\small $22.0$};
	\draw node at (2,\ilocationtimestamp) {\small $22.1$};
	\draw node at (3.5,\ilocationtimestamp) {\small $22.2$};
	\draw node at (5,\ilocationtimestamp) {\small $22.3$};
	\draw node at (6.5,\ilocationtimestamp) {\small $22.4$};
	\draw[dashed] (-2,-1) -- (7.2,-1);
	\draw node at (-1,\ilocationindex) 
	{\small \indexvariable};
	\draw node at (0.5,\ilocationindex) {\small $221$};
	\draw node at (2,\ilocationindex) {\small $222$};
	\draw node at (3.5,\ilocationindex) {\small $223$};
	\draw node at (5,\ilocationindex) {\small $224$};
	\draw node at (6.5,\ilocationindex) {\small $225$};
    \draw[dashed] (-2,-1.5) -- (7.2,-1.5);
	\pgfmathsetmacro{\llocation}{-1.5}
	\pgfmathsetmacro{\dlocation}{-2}
	\pgfmathsetmacro{\xlocation}{-2.5}
	\pgfmathsetmacro{\edgelocation}{-3.5} ;
 \draw[dashed] (-2,-1.5) -- (6.8,-1.5);
	\pgfmathsetmacro{\llocation}{-1.5}
	\pgfmathsetmacro{\dlocation}{-2}
	\pgfmathsetmacro{\xlocation}{-2.5}
	\pgfmathsetmacro{\edgelocation}{-3.5} ;
	\draw (2.7,0.6) -- (4.3,0.6) -- (4.3,-1.6) -- (2.7,-1.6) -- (2.7,0.6);
		\draw node at (3.9,-1.8) {\small Record $r_3$};
	\end{tikzpicture} }
    \caption{Fixed Step Trace.}
    \label{fig:fixedTrace}
    \end{subfigure}
    \begin{subfigure}[b]{\columnwidth}
    \scalebox{0.9}{
      \begin{tikzpicture}
	\pgfmathsetmacro{\ilocationangularrate}{0.25}
	\pgfmathsetmacro{\ilocationmode}{-0.25}
	\pgfmathsetmacro{\ilocationtimestamp}{-0.75}
	\pgfmathsetmacro{\ilocationindex}{-1.2}
	\draw[dashed] (-2,0.5) -- (7.2,0.5);
		\draw node at (-1,\ilocationangularrate) {\footnotesize $F\_s~[kg/s]$};
	\draw node at (0.5,\ilocationangularrate) {\small $4.006$};
	\draw node at (2,\ilocationangularrate) {\small $4.001$};
	\draw node at (3.5,\ilocationangularrate) {\small $3.999$};
	\draw node at (5,\ilocationangularrate) {\small $3.994$};
	\draw node at (6.5,\ilocationangularrate) {\small $3.989$};
	\draw[dashed] (-2,-0) -- (7.2,-0);
		\draw node at (-1,\ilocationmode) {\footnotesize $P\_s~[bar]$};
        \draw node at (0.5,\ilocationmode) {\small $87.321$};
	\draw node at (2,\ilocationmode) {\small $87.320$};
	\draw node at (3.5,\ilocationmode) {\small $87.319$};
	\draw node at (5,\ilocationmode) {\small $87.307$};
	\draw node at (6.5,\ilocationmode) {\small $87.285$};
		\draw[dashed] (-2,-0.5) -- (7.2,-0.5);
		\draw node at (-1,\ilocationtimestamp) {\small $\timevariableEncoding~[s]$};
	\draw node at (0.5,\ilocationtimestamp) {\small $22.05$};
	\draw node at (2,\ilocationtimestamp) {\small $22.32$};
	\draw node at (3.5,\ilocationtimestamp) {\small $22.40$};
	\draw node at (5,\ilocationtimestamp) {\small $22.67$};
	\draw node at (6.5,\ilocationtimestamp) {\small $22.94$};

	\draw[dashed] (-2,-1) -- (7.2,-1);
	\draw node at (-1,\ilocationindex) 
	{\small \indexvariable};
	\draw node at (0.5,\ilocationindex) {\small $103$};
	\draw node at (2,\ilocationindex) {\small $104$};
	\draw node at (3.5,\ilocationindex) {\small $105$};
	\draw node at (5,\ilocationindex) {\small $106$};
	\draw node at (6.5,\ilocationindex) {\small $107$};
    \draw[dashed] (-2,-1.5) -- (7.2,-1.5);
	\pgfmathsetmacro{\llocation}{-1.5}
	\pgfmathsetmacro{\dlocation}{-2}
	\pgfmathsetmacro{\xlocation}{-2.5}
	\pgfmathsetmacro{\edgelocation}{-3.5} ;
 \draw[dashed] (-2,-1.5) -- (7.2,-1.5);
	\pgfmathsetmacro{\llocation}{-1.5}
	\pgfmathsetmacro{\dlocation}{-2}
	\pgfmathsetmacro{\xlocation}{-2.5}
	\pgfmathsetmacro{\edgelocation}{-3.5} ;
	\draw (2.7,0.6) -- (4.3,0.6) -- (4.3,-1.6) -- (2.7,-1.6) -- (2.7,0.6);
		\draw node at (3.9,-1.8) {\small Record $r_3$};
	\end{tikzpicture} }
    \caption{Variable Step Trace.}   
    \label{fig:variableTrace}
 \end{subfigure}
\caption{\textcolor{blue}{Two examples of a trace for the example from \Cref{fig:signals}.}}
\label{fig:traceExample}
\end{figure*}

\begin{definition}[Trace]
{\rm
\textcolor{blue}{
A \textit{trace} \tracesymbol  is a tuple \trace. 
The input  $\inputInterpretation: \nonnegativenatural  \times \inputs \rightarrow   \domain $ and
 output $\outputInterpretation: \nonnegativenatural  \times \outputs \rightarrow   \domain $ interpretations
 associate each input $\inputVariable \in \inputs$ and output $\outputVariable \in \outputs$ variable and position $\indexvariable \in \nonnegativenatural$  with a value $\inputInterpretation(\indexvariable,\inputVariable)$ and $\outputInterpretation(\indexvariable,\outputVariable)$ from their domain $\type(\inputVariable)$ and $\type(\outputVariable)$.
  The timestamp interpretation $\interpretationtime: \nonnegativenatural \rightarrow  \real $ associates each position \indexvariable to a timestamp value.}}
$\hfill$ $\Box$
\end{definition}

\textcolor{blue}{
\Cref{fig:traceExample} presents two examples of (fragments of) traces for the example from \Cref{fig:signals}.
The traces show the values of the input ``Steam flow rate'' ($F\_s$) and the output ``Steam Pressure'' ($P\_s$) variables change over time.
The function $\interpretationtime$ associates each position with a timestamp. 
For example, for the trace in \Cref{fig:fixedTrace},  $\interpretationtime(223)=22.2$s and $\interpretationtime(225)=22.4$s.
The functions \inputInterpretation and \outputInterpretation associate a value for each timestamp position and variable. 
For example, for the trace in 
 \Cref{fig:fixedTrace}, $\inputInterpretation(223,$F\_s$)=4.003~\frac{kg}{s}$ and 
$\outputInterpretation(223,$P\_s$)=87.321~\mathit{bar}$.
The symbol \traces denotes the universe of all the traces.
A trace record refers to the values assumed by the timestamp and the input and output variables at a specific position of the trace.
 For example, the trace records $r_3$ of the two traces in \Cref{fig:traceExample} are delimited by a square frame.
When we refer to a generic variable \variable, that can refer both to an input or an output variable, we use $\interpretation{\variable}_\indexvariable$ to indicate its interpretation in position \indexvariable, that is $\interpretation{\variable}_\indexvariable=\inputInterpretation(\indexvariable,\variable)$ when \variable is an input variable, and 
$\interpretation{\variable}_\indexvariable=\outputInterpretation(\indexvariable,\variable)$ when \variable is an output variable.
For example,  for the trace in \Cref{fig:fixedTrace},  the interpretation of the variable $F\_s$ at position $223$ is  $\interpretation{$F\_s$}_{223}=\inputInterpretation(223,$F\_s$)=4.003~\frac{kg}{s}$.
}

\textcolor{blue}{
We consider two types of trace:  fixed and variable sample step traces. 
A fixed sample step trace records the signal values regularly after a specific sample step $T_s$.
For example, \Cref{fig:fixedTrace} has a fixed sample step: The trace records are recorded every $0.1$ time instants, as exemplified by the values assumed by the variable $\tau$ that collect the timestamp at which the signal values are recorded.
Unlike fixed sample step traces, the sample step of variable sample step traces changes.
For example, \Cref{fig:variableTrace} has a variable sample time as exemplified by the values assumed by the variable $\tau$.
}

\begin{definition}[Fixed and Variable Step Traces]
{\rm \textcolor{blue}{A trace \tracesymbol has a \emph{fixed sampling time} \sampleTime if for every 
index $\indexvariable\in \nonnegativenatural, 
 \interpretationtime(\indexvariable+1)-\interpretationtime(\indexvariable)=\sampleTime$. 
 A trace with a fixed sampling time is a \emph{fixed step trace}.
 Otherwise, the trace is a \emph{variable step trace} since it has a variable sampling time.
 }}
 $\hfill$ $\Box$
\end{definition}
\textcolor{blue}{
The trace from \Cref{fig:fixedTrace} has a fixed sampling time since for all index $\indexvariable\in \nonnegativenatural, \interpretationtime(\indexvariable+1)-\interpretationtime(\indexvariable)=0.1$.
The trace from \Cref{fig:variableTrace} has a variable sampling time since the difference ($0.08s$) between the timestamp ($22.32$) in position $104$ and the timestamp ($22.40$) in position $105$ differs from the difference ($0.27s$) between the timestamp ($22.40$) in position $105$ and the timestamp ($22.67$) in position $106$.
}

\change{
The semantics of Requirements Tables can be defined by considering fixed and variable step traces as detailed in the following section.
}{traceSectionEnd}{algorithm_completeness}

\begin{figure*}

\footnotesize
\begin{subfigure}[b]{\textwidth}
\begin{tabular}{p{0.95\textwidth}} 
\toprule
$\begin{aligned}
& \interpretation{\constant}_{\indexvariable ,\tracesymbol} & \coloneq & \constant & \nonumber\\ 
& \interpretation{\variable}_{\indexvariable ,\tracesymbol} & \coloneq & \interpretation{\variable}_{\indexvariable, \tracesymbol} & \nonumber\\ 
& \interpretation{\prev(\variable)}_{\indexvariable ,\tracesymbol} & \coloneq &  
\begin{cases}
   \interpretation{\variable}_{\indexvariable-1, \tracesymbol} & \text{if }\indexvariable>0   \\
   \interpretation{\variable}_{\indexvariable, \tracesymbol} & \text{if } \indexvariable=0   
\end{cases}
\nonumber\\ 
&  \interpretation{\term_1   \odot \term_2}_{\indexvariable ,\tracesymbol} & \coloneq & \interpretation{\interpretation{\term_1}_{\indexvariable ,\tracesymbol}\odot \interpretation{\term_2}_{\indexvariable ,\tracesymbol}}_{\indexvariable ,\tracesymbol} & \nonumber\\
& \interpretation{\term_1\oplus \term_2}_{\indexvariable, \tracesymbol} & \coloneq & 
\begin{cases} 
\interpretation{\term_2}_{\indexvariable ,\tracesymbol}-\interpretation{\term_1}_{\indexvariable ,\tracesymbol} &\text{if } \odot \in \{<, \leq \}  \\
\interpretation{\term_1}_{\indexvariable ,\tracesymbol}-\interpretation{\term_2}_{\indexvariable ,\tracesymbol} & \text{if } \odot \in \{>, \geq \} \\
abs(\interpretation{\term_1}_{\indexvariable ,\tracesymbol}-\interpretation{\term_2}_{\indexvariable ,\tracesymbol}) & \text{if } \odot \in \{=\} \\
\end{cases} \nonumber \\
&  \interpretation{\neg \logicalexpression}_{\indexvariable, \tracesymbol} & \coloneq & -\interpretation{\logicalexpression}_{\indexvariable, \tracesymbol} & \nonumber\\
&  \interpretation{\logicalexpression_1 \oslash \logicalexpression_2}_{\indexvariable, \tracesymbol}  & \coloneq & 
\begin{cases}
\text{if } \oslash=``\&'' & \text{then } min(\interpretation{\logicalexpression_1}_{\indexvariable, \tracesymbol}, \interpretation{\logicalexpression_2}_{\indexvariable, \tracesymbol})\\
\text{if } \oslash=``|'' & \text{then } max(\interpretation{\logicalexpression_1}_{\indexvariable, \tracesymbol}, \interpretation{\logicalexpression_2}_{\indexvariable, \tracesymbol})\\ 
\text{if } \oslash=``\Rightarrow''  & \text{then } \indexvariable, \tracesymbol \models (\neg \logicalexpression_1 \oslash \logicalexpression_2)\\
\end{cases}
 & \nonumber\\
& \interpretation{\preconditions  \Rightarrow \postconditions}_{\indexvariable, \tracesymbol}  & \coloneq & max(-\interpretation{\preconditions}_{\indexvariable, \tracesymbol}, \interpretation{\postconditions}_{\indexvariable, \tracesymbol}) & \nonumber\\  
& \interpretation{\requirement, \requirementTable}_{\indexvariable, \tracesymbol} & \coloneq & min(\interpretation{\requirement}_{\indexvariable, \tracesymbol},  \interpretation{\requirementTable}_{\indexvariable, \tracesymbol} )  & \nonumber 
\end{aligned}$\\
\bottomrule
abs: absolute value
\end{tabular}
\caption{Semantics for the operators not based on timestamp values.}
\label{fig:commonOperators}
\end{subfigure}
\footnotesize
\begin{subfigure}[b]{\textwidth}
\begin{tabular}{p{0.95\textwidth}} 
\toprule
$\begin{aligned}
& \interpretation{\duration(\logicalexpression) \geq \constant}_{\indexvariable, \tracesymbol} & \coloneq & min_{\interpretationtime(\indexvariable)\geq \constant \wedge (k\in \positivenatural.(\interpretationtime(\indexvariable)-\interpretationtime(\indexvariable-k) \leq c_r))} \interpretation{\logicalexpression}_{k, \tracesymbol} \nonumber\\
& \interpretation{\preconditions[\durationcolumn]  \Rightarrow \postconditions}_{\indexvariable, \tracesymbol} & \coloneq & max((-min_{\interpretationtime(\indexvariable)\geq \constant, k\in \positivenatural.(\interpretationtime(\indexvariable)-\interpretationtime(\indexvariable-k) \leq d_r)} \interpretation{\preconditions}_{k, \tracesymbol}), 
\interpretation{\postconditions}_{\indexvariable, \tracesymbol}) \nonumber\\
\end{aligned}$\\
\bottomrule
$^\ast$ $c_r=\lceil \frac{\constant}{\sampleTime} \rceil \cdot \sampleTime$ and $d_r=\lceil \frac{\durationcolumn} {\sampleTime} \rceil \cdot \durationcolumn$. The operators are defined for $\constant \geq \sampleTime$ and $\durationcolumn \geq \sampleTime$. 
\end{tabular}
\caption{Fixed Step Semantics.}
\label{fig:fixedStepSemantics}
\end{subfigure}
\begin{subfigure}[b]{\textwidth}
\begin{tabular}{p{0.95\textwidth}} 
\toprule
$\begin{aligned}
&  \interpretation{\duration(\logicalexpression) \geq c}_{\indexvariable, \tracesymbol}
& \coloneq & max_{k \in \positivenatural.( k\leq\indexvariable \wedge \interpretationtime(\indexvariable)-\interpretationtime(k)\geq c)} min_{j \in \positivenatural.(k \leq j \leq \indexvariable)} \interpretation{\logicalexpression}_{j, \tracesymbol}  & \nonumber\\
& \interpretation{\preconditions[\durationcolumn]  \Rightarrow \postconditions}_{\indexvariable, \tracesymbol}  & \coloneq & max(-(max_{k \in \positivenatural.( k\leq\indexvariable \wedge \interpretationtime(\indexvariable)-\interpretationtime(k) \geq \durationcolumn)}  min_{j \in \positivenatural.(k \leq j \leq \indexvariable)} \interpretation{\preconditions}_{j, \tracesymbol} ), \interpretation{\postconditions}_{\indexvariable, \tracesymbol}) \nonumber\\
\end{aligned}$\\
\bottomrule
\end{tabular}
\caption{Variable Step Semantics.}
\label{fig:variableStepSemantics}
\end{subfigure}
     \caption{Requirements Tables: Quantitative Semantics.}
    \label{fig:semantics}
\end{figure*}

\subsection{Semantics}
\label{sec:semantics}
\change{
We define a quantitative trace-based semantics of Requirements Tables by first defining (Definition~\ref{def:position}) a metric for the satisfaction of (the requirements of) a Requirements Table (\requirementTable) in a position (\indexvariable) of a trace (\tracesymbol), and then by defining (Definition~\ref{def:traceSatisfaction}) a metric for the satisfaction of a Requirements Table (\requirementTable) in that trace (\tracesymbol).}{semanticsSectionStart}{algorithm_completeness}

\begin{definition}[Semantics]
{\rm
\textcolor{blue}{
Let \requirementTable,  $\tracesymbol=\trace$, and \indexvariable$ \geq 0$ respectively be a Requirements Table, a trace, and a position.
The satisfaction $\interpretation{\requirementTable}_{\indexvariable, \tracesymbol}$ of the Requirements Table \requirementTable
in position \indexvariable of the trace $\tracesymbol$ is recursively defined in \Cref{fig:semantics}.}}
\label{def:position}
\end{definition}

\textcolor{blue}{
The semantics (a.k.a. interpretation) from \Cref{fig:semantics} includes all the different operators of the Requirements Tables.
For the operator $\duration(\logicalexpression)$ and the requirement expressed using the duration column (\durationcolumn), which requires considering the values assumed by the timestamp values, we defined two semantics: the fixed (\Cref{fig:fixedStepSemantics}) and the variable  (\Cref{fig:variableStepSemantics}) step semantics that respectively refer to fixed and variable step traces.
In the following, we describe the semantics of each operator.}

\textcolor{blue}{
For a position \indexvariable, the value $\interpretation{\constant}_{\indexvariable, \tracesymbol}$ of the constant \constant corresponds to its value.}

\textcolor{blue}{
The interpretation $\interpretation{\variable}_{\indexvariable, \tracesymbol}$ of a variable \variable at position \indexvariable is value $\interpretation{\variable}_\indexvariable$ (see \Cref{sec:traces}). 
For example, for \Cref{fig:fixedTrace}, the interpretation of the variable $F\_s$ at position $223$ is  $\interpretation{F\_s}_{223}=4.003~\frac{kg}{s}$.}

\textcolor{blue}{
The interpretation of the term $\prev(\variable)$ at position \indexvariable is the value 
$\interpretation{\variable}_{\indexvariable-1}$ assumed by the variable \variable in the previous position $\indexvariable-1$ of the trace, or the initial value ($\interpretation{\variable}_0$) assigned to the variable (\variable) if the position is $0$. 
For simplicity, in this work, we assume that every variable is assigned to an initial value. 
For example, for \Cref{fig:fixedTrace}, the interpretation $\prev(F\_s)$ of the variable $F\_s$ at position $223$ is $\interpretation{F\_s}_{222}=4.005~\frac{kg}{s}$. }

\textcolor{blue}{
The interpretation of $\term_1  \odot \term_2$ is obtained by applying the arithmetic operator $\odot$ to the interpretations of $\term_1$ and $\term_2$.
For example, for \Cref{fig:fixedTrace}, the interpretation of $F\_s$+$5$ at position $223$ is $9.003$, that is, the sum between $5$ and the interpretation of the input variable $F\_s$ (i.e., $4.003$) at position~$223$.}

\textcolor{blue}{
The interpretation of an expression of the form 
 $\term_1\oplus \term_2$ quantifies its satisfaction degree at position \indexvariable, i.e., it quantifies how $\term_1$ and $\term_2$ satisfy the relation specified by the relational operator $\oplus$.
 For example, for \Cref{fig:fixedTrace}, the interpretation of $F\_s \geq 4$ at position $223$ is $0.003$, that is, the difference between the interpretation of the input variable $F\_s$ (i.e., $4.003$) at position~$223$ and the value $4$.}

\textcolor{blue}{
The semantics of an expression of the form $\neg \logicalexpression$ is the opposite of the satisfaction of the expression $ \logicalexpression$ at position \indexvariable.
 For example, for \Cref{fig:fixedTrace}, the quantitative semantics of the expression $\neg(F\_s>4)$ is $-0.003$.}

\textcolor{blue}{
The semantics of the logical expression $\logicalexpression_1 \oslash \logicalexpression_2$ 
quantifies its satisfaction degree at position \indexvariable, i.e., it quantifies how $\term_1$ and $\term_2$ satisfy the relation specified by the logical operator $\oslash$.
 For example, for \Cref{fig:fixedTrace}, for the logical expression $(P\_s>87) \& (P\_s<87.5)$ at position~$223$, the quantitative semantics is obtained by computing the minimum ($0.179$) between the quantitative value ($0.321$) of the expression $P\_s>87$ and the quantitative value ($0.179$) for the expression $P\_s<87.5$.}

\textcolor{blue}{
The semantics of $\preconditions  \Rightarrow \postconditions$ quantifies its satisfaction degree at position \indexvariable, i.e., it computes the maximum between the opposite of the satisfaction value of the precondition (\preconditions) and the satisfaction value of the postcondition (\postconditions).}

\textcolor{blue}{
The semantics of $\requirement, \requirementTable$ specifies that the satisfaction value is the minimum between the satisfaction value for the requirement $\requirement$ and the other requirements of the Requirements Table $\requirementTable$. }

\textcolor{blue}{
Two semantics for the operator $\duration(\logicalexpression) \geq \constant$ and requirements expressed using a duration column are defined: Fixed step and variable step semantics. These semantics respectively support fixed and variable step traces.}

\begin{itemize}[leftmargin=0.07in]
    \item \emph{Fixed step semantics}. \textcolor{blue}{For the operator $\duration(\logicalexpression) \geq \constant$, we compute the minimum across the values assumed by the logical expression $\logicalexpression$ for a time window with duration $c$.
The value $c_r=\lceil \frac{\constant} {\sampleTime} \rceil \cdot \sampleTime$ is used to define the operator's semantics, where $\lceil \cdot \rceil$ is the ceil function.
This value accounts for the value $\constant$ not being a multiplier of the fixed sampling time $\sampleTime$. The interested reader can consult~\cite{menghi2024completeness} for additional details.\\
The fixed step semantics of a requirement expressed using a duration column is the maximum among the opposite of the satisfaction value computed for the precondition (by considering the duration $d$), and its postcondition. }
\item \emph{Variable step semantics}.  \textcolor{blue}{
The variable step semantics for the operator $\duration(\logicalexpression) \geq \constant$ and requirements expressed using a duration column correspond to the one proposed for the fixed step semantics. However, being the trace based on a variable step, the \emph{max} operators are used to identify the intervals associated with the expression $\duration(\logicalexpression) \geq \constant$ and the duration column to be considered for the computation of the quantitative measures. }
\end{itemize}

\begin{definition}[Satisfiability]
\label{def:traceSatisfaction}
\textcolor{blue}{
A trace \tracesymbol satisfies a Requirements Table \requirementTable,  if  $\interpretation{\requirementTable}_{\indexvariable, \tracesymbol}>0$  for every position \indexvariable. We use the notation $ \interpretation{\requirementTable}_\tracesymbol$ to indicate the minimum value of $\interpretation{\requirementTable}_{\indexvariable, \tracesymbol}$ for every position $\indexvariable$ of the trace.}
 \end{definition}

\textcolor{blue}{
A trace (\tracesymbol) satisfies (the requirements of) a Requirements Table (\requirementTable) if the value $\interpretation{\requirementTable}_\tracesymbol$ is greater than $0$.}

\begin{theorem}
\textcolor{blue}{
Our semantics satisfies the \emph{well-formedness} properties from Definition~\ref{def:ftprop}}
\end{theorem}
\begin{proof}[Proof Sketch] 
\textcolor{blue}{
The definition of the quantitative semantics for the relational operators (\Cref{fig:semantics}) ensures that (a)~negative and positive values indicate that the relation is violated and satisfied,  (b)~the higher the positive value, the more it is satisfied, and (c)~the lower the positive value, the more it is violated.}

\textcolor{blue}{
The definition of the quantitative semantics for the Boolean operators, requirements, and RT (\Cref{fig:semantics}) preserves this property.}
\end{proof}

\change{
Our fitness metric (and translation from RT to Stateflow) is defined by considering the quantitative semantics of Requirements Tables from Definition~\ref{def:traceSatisfaction}.}{semanticsSectionEnd}{algorithm_completeness}

 \section{From Requirements Table to Stateflow}
\label{sec:transl}

This section describes the behavior of our \emph{Fitness-Function Generator} (\phase{1})
translating RT into Stateflow models that compute a fitness function that drives our SBST procedure.
First, we describe the behavior of our procedure for RTs that \emph{do not} use the duration and the previous operators (\Cref{sec:algorithm}).
Then, in \Cref{sec:dur} and \Cref{sec:prev}, we describe how our procedure supports these operators.
Finally, \Cref{sec:soundness} proves that our fitness function satisfies the properties from \Cref{def:ftprop}.

To illustrate our procedure, we use the RT from \Cref{fig:exampleRT} and the corresponding fitness function from \Cref{fig:FitCalculator} consisting of the Stateflow model from \Cref{fig:exampleSF} and the aggregator from \Cref{fig:FitConverter}.

\subsection{The \RTtoStateflow Algorithm}
\label{sec:algorithm}

\Cref{alg:RT2Stateflow} contains the pseudocode of our Fitness-Function Generator (\RTtoStateflow).
The input of our algorithm is an RT; the output is a fitness function (FF).

\begin{itemize}
\item Line~\ref{alg:ffCreation} creates the main structure of the fitness function, as detailed in \Cref{fig:FitCalculator}.
The fitness function consists of two parts: The Stateflow model and the aggregator.
Then, the algorithm starts constructing the Stateflow subcomponent (\Cref{fig:exampleSF}).
    \item Line~\ref{alg:initialization} creates an empty Stateflow.
\item Line~\ref{alg:initializeRT} creates a set of parallel states, one for each requirement of the RT.
For the running example, our procedure generates the parallel states R1, R2, and R3 from \Cref{fig:exampleSF} associated with the Requirements 1, 2, and 3 of the RT from \Cref{fig:exampleRT}.
\item Line~\ref{alg:prepost} creates the states PRC and POA in each of the parallel states.
The WT state will be described in \Cref{sec:dur}.
The Stateflow will be in the state PRC (precondition checking) when it is checking for the satisfaction of the precondition, it will move to the state POA (postcondition active) when the precondition is satisfied.
\item Line~\ref{alg:labelTransitions} adds the transitions from state PRC to POA and from POA to PRC. 
The algorithm labels the transition from PRC to POA of the Stateflow with the Boolean expression associated with the Precondition of the RT.
For example, the transition from PRC to POA of state R2 of the Stateflow model from \Cref{fig:exampleSF} is labeled with the Boolean expression ``$t>=30~\&~t<=35$'', i.e., the precondition of the Requirement~2 of the RT from \Cref{fig:exampleRT}.
The algorithm labels the transition from POA to PRC with the negation of the precondition, i.e., when the precondition is violated the Stateflow exits the state POA and returns to the state PRC.
\item Line~\ref{alg:addInit} adds the entry point of the state machine.
\item Line~\ref{alg:addJunction} adds the connective junction node.
The connective junction node is used to move to the POA state (that assesses the postcondition at the initial time instant) if the precondition is initially satisfied, to avoid entering the PRC state: An RT can perform one state-to-state transition at each timestamp.
\item Line~\ref{alg:addFitness} adds a fitness variable $\mathit{ff}1$, $\mathit{ff}2$, $\ldots$, $\mathit{ff}n$ associated with each requirement.
For our running example, it adds the fitness variables  $\mathit{ff}1$, $\mathit{ff}2$, and $\mathit{ff}3$ associated with the Requirements~1, 2, and 3.\\ 
It also adds an action that sets the value of the fitness variable for each state.
When the state machine is in the PRC, the value is set to infinite (Inf), e.g., see the action of the PRC state of the state machine R2.
When the state machine is in the POA, the fitness value is computed by considering the postcondition of the requirement associated with the state machine. Relational expressions in the form $\term_1\oplus \term_2$ with $\oplus \in \{>,\geq\}$ is converted into the arithmetic expression ~$\term_1- \term_2$. 
\rep{Relational expressions in the form $\term_1\oplus \term_2$ with $\oplus \in \{<,\leq\}$ are first converted into equivalent expressions in the form $\term_1^\prime\oplus \term_2^\prime$ with $\oplus \in \{>,\geq\}$ using standard algebraic operations}{We omit the translation for the other relational operators since they can be derived from these relational operators by using the Boolean operators.}{RelOpConversion}{small_contribution_1}
Boolean expressions in the form 
$\logicalexpression_1 \oslash \logicalexpression_2$ 
are converted into arithmetic expressions by converting the logical expressions $\logicalexpression_1$ and $\logicalexpression_2$ into arithmetic expressions and by replacing the logical operators ``$\&$'' and ``$|$'' with the operators \emph{min} and \emph{max}.
Boolean expressions in the form $\neg \logicalexpression$ are converted into the opposite of the arithmetic expression (the arithmetic expression preceded by a minus sign) computed from the logical expression $\logicalexpression$.
For example, the postcondition $\mathit{P\_s}>87~\&~\mathit{P\_s}<87.5$ of the requirement~2 is converted into the arithmetic expression $min(\mathit{P\_s}-87,-(\mathit{P\_s}-87.5))$, where $\mathit{P\_s}-87$ is obtained from $\mathit{P\_s}>87$,  $-(\mathit{P\_s}-87.5)$ is obtained from $\mathit{P\_s} \leq 87.5$.
This instruction terminates the generation of the Stateflow model from \Cref{fig:FitCalculator}.
\item Line~\ref{alg:addActions} adds the content of the Action columns to the POA state of the corresponding requirement.
\item Line~\ref{alg:min} generates the aggregator component detailed in \Cref{fig:FitConverter}.
The aggregator component computes the fitness measure from the fitness measures generated by each state machine.
The \texttt{min} Simulink block computes the minimum value assumed by the signals of its input ports.
For example, 
the  \texttt{min} block from \Cref{fig:FitConverter} is connected to the input signals $\mathit{ff}1$, $\mathit{ff}2$, and $\mathit{ff}3$. 
The feedback loop with a delay block (to capture the minimum value at the previous time) computes the minimum fitness value over the entire simulation.
\item Line~\ref{alg:return} returns the fitness function.
The fitness value is the value assumed by the signal $\mathit{ff}\_\mathit{Total}$ at the end of the simulation.
\end{itemize}

\changeRead{After the fitness function is generated from the RT it is plugged into the model. The input signals of the Fitness Function are linked to the corresponding signals in the model as detailed in \Cref{fig:exampleModel}.
For example, the input signal $F\_s$ of the RT is connected to the corresponding signal from the output of the Steam Condenser Model and Controller.}{stateflow_model_connectionAnswer}{stateflow_model_connection}

\change{As the test case is executed, the fitness function is evaluated since it is an additional block of the model. The fitness value is the last value assumed by the fitness function during the simulation.}{stateflow_model_connectionAnswerb}{stateflow_model_connection}.

\begin{algorithm}[t]
\footnotesize
\caption{The \RTtoStateflow algorithm.}
\label{alg:RT2Stateflow}
\begin{algorithmic}[1]
\Function{\RTtoStateflow}{\RT}\\
 \codespace \label{alg:ffCreation}\textsc{FF}=\texttt{createFitnessFunction}();\\
 \label{alg:initialization}  \codespace SF=\texttt{createEmptyStateflow}(\textsc{FF});\\
 \codespace\label{alg:initializeRT}\texttt{initializeRT}(\RT,SF);\\
 \label{alg:prepost} \codespace \texttt{createPrePost}(SF);\\
 \label{alg:labelTransitions} \codespace \texttt{createTransitions}(\RT,SF);\\
 \label{alg:addInit} \codespace \texttt{addInit}(SF);\\
 \label{alg:addJunction} \codespace \texttt{addConnectiveJunction}(SF);\\
 \label{alg:addFitness} \codespace  \texttt{addFitness}(\textsc{SF});\\
 \label{alg:addActions} \codespace  \texttt{addActions}(\textsc{SF});\\
 \label{alg:min} \codespace  \texttt{addMinCalculation}(\textsc{FF});\\
 \label{alg:return} \codespace \textbf{return} FF;
\EndFunction\
\end{algorithmic}
\end{algorithm}

\subsection{Supporting RTs with Durations}
\label{sec:dur}
To support requirements using the \duration operator, the procedure from \Cref{alg:RT2Stateflow} is modified as follows. 
\begin{itemize}
    \item Line~\ref{alg:prepost} creates the WT state in addition to states PRC and POA in each of the parallel states that refer to requirements that specify a duration.
For example, the WT state is added to the state machine R3 from \Cref{fig:exampleSF}, since the duration is specified for Requirement 3 from \Cref{fig:exampleRT}. 
\item Line~\ref{alg:labelTransitions} adds two transitions to the WT state from the PRC state and the connective junction. The state machine transitions from the PRC to the WT state when the precondition is satisfied, and from the connective junction to the WT state if the precondition is satisfied when the simulation starts.
For example, the Stateflow from \Cref{fig:exampleSF} transitions from PRC to WT and from the connective junction to the WT when the precondition $\mathit{F\_s}\geq4$ of Requirement~3 from \Cref{fig:exampleRT} is satisfied.
The state machine remains in the WT state for the time specified by the duration value if the precondition remains satisfied.
If the precondition becomes unsatisfied before the duration threshold is reached, the state machine returns to the PRC state.
For example, the Stateflow from \Cref{fig:exampleSF} transitions from WT to PRC when the precondition $\mathit{F\_s}\geq4$ of Requirement~3 from \Cref{fig:exampleRT} is violated.
Otherwise, when the elapsed time (encoded by the $\mathit{et}$ variable) exceeds the duration value, the state machine moves to the POA state.
For example, the Stateflow from \Cref{fig:exampleSF} remains in the WT until the value of the elapsed time variable (\texttt{et}) matches the value $5$ specified by the duration column of the RT from \Cref{fig:exampleRT}.
\end{itemize}

\subsection{Supporting RTs Containing the Previous Operator}
\label{sec:prev}
To support requirements using the \prev operator, the procedure from \Cref{alg:RT2Stateflow} is modified as follows. 
\begin{itemize}
    \item Line~\ref{alg:initialization} analyzes each requirement of the RT and extracts expressions in the form $\prev(s)$, where $s$ is the name of the signal scoped by the \prev operator.
For example, for the RT from \Cref{fig:exampleRT}, it extracts the expression $\prev(F\_s)$.
It then analyzes each expression in the form $\prev(s)$ and adds an input signal to the Stateflow named $\mathit{sPrev}$. 
The signal $s$ is connected with a Simulink delay block.
The output of the delay is then connected to the input $\mathit{sPrev}$ of the Stateflow.
For example, from our RT our translation added an input signal to the Stateflow from \Cref{fig:FitCalculator} named $\mathit{F\_sPrev}$ where $\mathit{F\_s}$ is the name of the signal scoped by the \prev operator. 
The signal $\mathit{F\_s}$ is connected with a Simulink delay block. The output of the delay is then connected to the input $\mathit{F\_sPrev}$ of the Stateflow.
The expression $\prev(\mathit{F\_s})$ of the RT is replaced with the expression 
 $\mathit{F\_sPrev}$. 
\end{itemize}
Then, the translation proceeds as described in \Cref{sec:transl}.

\subsection{Soundness and Completeness}
\label{sec:soundness}
We prove that the fitness function generated from the algorithm described in \Cref{alg:RT2Stateflow} satisfies the properties from \Cref{def:ftprop}. Formally,
\begin{theorem}
\label{sec:soundTheorem}
Let $\mathit{rt}$ be an $\mathit{RT}$ and $\mathit{ff}$ the fitness function generated from $\mathit{rt}$ by the \RTtoStateflow algorithm, then $\mathit{ff}$ satisfies the properties from \Cref{def:ftprop}.
\end{theorem}
\begin{proof}[Proof Sketch] Our proof relies on three statements (stmt).

(a) For each postcondition, converting relational expressions in the form $\term_1\oplus \term_2$ with $\oplus \in \{>,\geq\}$ into arithmetic expressions of the form~$\term_1- \term_2$, and replacing the logical operators ``$\&$'', ``$|$'', ``$\neg$'' with the operators \emph{min}, \emph{max}, and ``$-$'' returns a value that~is positive if the postcondition is satisfied and negative otherwise,  the higher the positive fitness value is, the farther the system is from violating the postcondition; and the lower the negative value is, the farther the system is from satisfying the postcondition~\cite{FainekosPappas2009,menghi2019generating,larsen1988modal}. This consideration is consistent with the semantics from \Cref{sec:semantics}.

(b) Our translation ensures that the postcondition is assessed only when the precondition is satisfied: The state machine enters the $\mathit{POA}$ state when the condition specified by the precondition is satisfied (for the specified duration). This consideration is consistent with the semantics from \Cref{sec:semantics}, if the precondition is not satisfied, then the requirement is satisfied, i.e., $max(-\interpretation{\preconditions}_{\indexvariable, \tracesymbol}, \interpretation{\postconditions}_{\indexvariable, \tracesymbol})$ leads to a positive value.

(c) The aggregator component (\Cref{fig:FitConverter}) computes the minimum across the fitness values of the different requirements for the entire duration of the simulation.
This consideration is consistent with the semantics from Definition~\ref{sec:semantics} which defines  $ \interpretation{\requirementTable}_\tracesymbol$ as the minimum value of $\interpretation{\requirementTable}_{\indexvariable, \tracesymbol}$ for every position $\indexvariable$ of the trace.

Therefore, condition~1 of \Cref{def:ftprop} is satisfied: If at least one requirement is violated in at least one time instant, the aggregator returns a negative fitness value (stmt c) associated with the postcondition of the requirement that is violated (stmts a,b).
Condition~2 of \Cref{def:ftprop} is satisfied: If all the requirements are satisfied across the entire simulation, a positive fitness value is computed for all the postconditions (stmt a,b) and the aggregator returns a positive fitness value (stmt c).
Condition~3 of \Cref{def:ftprop} is satisfied:
When the requirements are satisfied, the aggregator returns the minimum across the fitness value of the different requirements (stmt c), that is the value of the fitness closer to a violation.
The farther this requirement from violation, the higher the positive fitness value returned (stmt a,b). 
Condition~4 of \Cref{def:ftprop} is satisfied: 
When the requirements are violated, the aggregator returns the minimum across the fitness value of the different requirements (stmt c), that is the value of the fitness farther from satisfaction.
The closer this requirement is to being satisfied, the higher the lower the fitness value returned (stmt a,b). 
\end{proof}

\begin{theorem}
\label{sec:soundTheoremSemantics}
\change{Let $\mathit{rt}$ be an $\mathit{RT}$ and $\mathit{ff}$ the fitness function generated from $\mathit{rt}$ by the \RTtoStateflow algorithm, then $\mathit{ff}$ computes the quantitative value defined by the semantics of the RT from \Cref{sec:Semantics}.}{theoremSemantics}{small_contribution_1}
\end{theorem}

\begin{proof}[Proof Sketch]
The proof reflects the one presented for Theorem~\ref{sec:soundTheorem}.
\end{proof}

\change{The syntax of Requirements Tables (from \cite{menghi2024completeness}) is reported in \Cref{sec:background}. 
The quantitative semantics is reported in \Cref{sec:Semantics}. 
Our solution is complete: All the operators are considered by our quantitative semantics and our translation.}{algorithm_completenessAnswerb}{algorithm_completeness} \section{Evaluation}
\label{sec:eval}
To evaluate our solution, we considered the following research questions (RQs):
\begin{itemize}[leftmargin=*]
    \item \textbf{RQ1}: How \emph{effective} is our SBST framework in generating failure-revealing test cases? (\Cref{sec:effectiveness})
    \item \textbf{RQ2}: How \emph{efficient} is our SBST framework in finding these test cases? (\Cref{sec:efficiency})
\end{itemize}
We describe our benchmark models and RT (\Cref{sec:benchmark}) and then answer our RQs.
Finally, we reflect on our results and present threats to Validity (\Cref{sec:discussion}).

\subsection{Benchmark and Tool Configuration}
\label{sec:benchmark}
A benchmark for SBST with RT does not exist, since RTs were created only recently (released at the beginning of 2022 \cite{RequirementsTable}).
Therefore, we considered \modelrtcombinations ($24+24+12$) model-RT combinations from three practical examples: the Automatic Transmission Controller, the Cruise Controller, and the Observer Mode Model models.
\change{\Cref{tab:ModelDescription} reports a brief description of these models, as well as some information about their complexity; i.e., the number of input and output signals, and the number of Simulink blocks in the model.  Note that the Cruise Controller Model is developed for the VI-CarRealTime~\cite{vicar} industrial simulator and contains 2714 blocks.}{ModelTable}{model_information}

\begin{table*}[t]
    \centering
    \caption{Description of the three models used in the Evaluation. For each model, we report the number of input (\textbf{\#In}) and output (\textbf{\#Out}) signals, as well as the number of Simulink blocks in the model (\textbf{\#Blocks}) and a brief description.}
    \footnotesize
    \begin{tabular}{l r r r p{12cm}}
        \toprule
        \textbf{Model}  &\textbf{\#In}   &\textbf{\#Out}    &\textbf{\#Blocks}  &\textbf{Description}\\
        \midrule
        AT      & 2 & 2 &  142  & Controller for an Automatic Transmission with 4 gears for a standard automotive powertrain. \\
        CC      & 2 & 4 & 2714  & Cruise Control for a 4-wheel drive Sedan with variable road slope. \\
        OMM     & 2 & 2 &   14  & Observer Mode Model from MathWorks for monitoring the behavior of a black-box system. \\
        \bottomrule
    \end{tabular}
    \label{tab:ModelDescription}
\end{table*}

\emph{Automatic Transmission Controller (AT)}. We considered $24$ model-RT combinations from the AT example~\cite{ARCH14}.
The AT comes from the Applied Verification for Continuous and Hybrid Systems (ARCH) competition~\cite{Arch2024}. 
This competition compares existing tools for testing CPS. 
We selected AT since it has the highest number (ten) of requirements, which are reasonably complex to violate~\cite{Arch2024}.

\begin{table*}[t]
    \centering
    \caption{Description of each version of the models of our benchmark.}
    \footnotesize
    \begin{tabular}{l p{13cm}}
        \toprule
        \textbf{MID}   &\textbf{Description}\\
        \midrule
        AT-v0   &Original model.\\
        AT-v1   &Increased the Engine Torque by 10\%.\\
        AT-v2   &Reduced vehicle inertia by 17\%.\\
        AT-v3   &Reduced torque acting on the Transmission by 25\%.\\
        \midrule
        CC-v6.1 & Earliest version of this model considered in this work.\\
        CC-v6.2 & Increased the smoothing on the output signals.\\
        CC-v7.1 & Lowered the maximum braking torque to avoid excessive negative accelerations.\\
        CC-v7.3 & Added component that periodically resets the integrator in the PID controller.\\
        CC-v7.4 & Added one more step transition to the desired velocity.\\
        CC-v7.5 & Added non-constant slope to the simulation.\\
        \midrule
        OMM-v0 & Original model from MathWorks\\
        OMM-v1 & Added cross contamination of 0.01 gain from input 1 to output 2.\\
        OMM-v2 & Added cross contamination of 0.01 gain from input 1 to output 2 and 0.01 gain from input 2 to output 1.\\
        OMM-v3 & Added cross contamination of 0.01 gain from input 1 to output 2 and 0.1 gain from input 2 to output 1.\\
        \bottomrule
    \end{tabular}
    \label{tab:ATmodels}
\end{table*}

We considered four versions of the AT model (AT-v0, AT-v1, \ldots, AT-v3): the original model (AT-v0) and three versions obtained by applying the changes to its parameters detailed in \Cref{tab:ATmodels} (top part).

We manually designed the RT for the requirements of AT since they are formalized in Signal Temporal Logic (STL)~\cite{maler2004monitoring}.  
We selected five requirements (AT1, AT2, AT6a, AT6b, AT6c) among the ten requirements of AT. 
We did not consider the requirement AT6abc since it is the conjunction of the requirements AT6a, AT6b, and AT6c that we are already adding to the RT.
We also did not consider the requirements AT51, AT52, AT53, and AT54 since they are easier to violate than the other requirements for most of the tools participating in the ARCH competition and contain nested temporal operators, which require modifying the model to encode them into RT. 
\Cref{tab:requirements} (top part) provides a textual description for each requirement.
We generated an RT modeling these five requirements.
Then, we parameterized each requirement. For example, AT1 can be instantiated by considering different values for the maximum vehicle speed (\textbf{SL1}) and the time limit that should be applied (\textbf{TL1}).
We considered six instances (AT-RT0, AT-RT1, \ldots, AT-RT5) of this RT generated by assigning the values from \Cref{tab:ATreqs} to its parameters.
Each row represents one instance of the RT and contains the value assigned to its parameters.

Therefore,  we considered $24$ ($4\times 6$) model-RT combinations in total, each made by a model and RT.
The combination $\langle$AT-v0, AT-RT0$\rangle$ represents the original model and requirement.

\emph{Cruise Controller (CC)}. We considered $24$ model-RT combinations from the CC example~\cite{10.1145/3611643.3613894}.
The CC model represents a model of a cruise controller for a four-wheel vehicle system designed in Simulink tested on the VI-CarRealTime~\cite{vicar} industrial simulator. 
The CC model was designed by producing 21 intermediate model versions and comes with four requirements.

\begin{table*}[t]
    \centering
    \caption{Requirements identifiers and description for our benchmark models. 
    The parameters for the requirements are highlighted in bold.}
        \label{tab:requirements}
        \footnotesize
    \begin{tabular}{l p{12.5cm}}
        \toprule
        \textbf{RqID} &\textbf{Description}\\
        \midrule
        AT1     &The vehicle speed shall be lower than \textbf{SL1} \(\mathit{mph}\) within the first \textbf{TL1} seconds.\\
        AT2     &The engine speed shall be lower than \textbf{RPM2} \(\mathit{rpm}\) within the first \textbf{TL2} seconds.\\
        AT6a    &If the engine speed is lower than \textbf{RPM6a} \(\mathit{rpm}\) within the first 30 seconds, then the vehicle speed shall be lower than \textbf{SL6a} \(\mathit{mph}\) within the first \textbf{TL6a} seconds.\\
        AT6b    &If the engine speed is lower than \textbf{RPM6b} \(\mathit{rpm}\) within the first 30 seconds, then the vehicle speed shall be lower than \textbf{SL6b} \(\mathit{mph}\) within the first \textbf{TL6b} seconds.\\
        AT6c    &If the engine speed is lower than \textbf{RPM6c} \(\mathit{rpm}\) within the first 30 seconds, then the vehicle speed shall be lower than \textbf{SL6c} \(\mathit{mph}\) within the first \textbf{TL6c} seconds.\\
\midrule
        CC-F1   &30 seconds after the desired velocity changes the difference between the desired velocity and actual velocity should reach and stay within \textbf{VT} \(\mathit{km}/h\).\\
        CC-D1   &After a grace period of \textbf{GT} seconds from the last driver input, the longitude acceleration should stay within -3 \(\mathit{m}/s^2\) and 5 \(\mathit{m}/s^2\).\\
        CC-D2   &After a grace period of \textbf{GT} seconds from the last driver input, the jerk should stay within -10 \(\mathit{m}/s^3\) and 10 \(\mathit{m}/s^3\).\\
        CC-D3   &After a grace period of \textbf{GT} seconds from the last driver input, the pitch acceleration should stay within -3 \(\mathit{rad}/s^2\) and 3 \(\mathit{rad}/s^2\).\\
        \midrule
        OMM1 &Output1 should be greater than \textbf{OL} when Input1 is greater than \textbf{IL}.\\
        OMM2 &Output2 should be greater than \textbf{OL} when Input2 is greater than \textbf{IL}.\\
        \bottomrule
    \end{tabular}
\end{table*}

\begin{table*}[t]
    \centering
    \caption{Requirements Tables from our benchmark.
    Each row represents one instance of the RT and contains the value assigned to its parameters.}
    \label{tab:ATreqs}
    \footnotesize
    \begin{tabular}{l  r r r r r r r r r r r r r}
        \toprule
        \textbf{ReqID} &\multicolumn{13}{c}{\textbf{Requirements Parameter}}\\
        \midrule
        & \makecell{SL1} & TL1 &  RPM2 &  TL2  &  RPM6a & SL6a & TL6a &  RPM6b   & SL6b & TL6b  &  RPM6c   & SL6c & TL6c \\
AT-RT0 & 120 & 20 & 4750 & 10 & 300 & 35 & 4 & 3000 & 50 & 8 & 300 & 65 & 20 \\ 
        AT-RT1 & 115 & 20 & 4750 & 10 & 2900 & 32 & 4 & 3000  & 50 & 8 & 3000 & 65 & 20 \\
        AT-RT2  & 115 & 20 & 4800 & 10 & 3000 & 35 & 4 & 2900 & 52 & 8 & 3000 & 65 & 20   \\
        AT-RT3  & 125 & 25 & 4800 & 10 & 3000 & 35 & 4 & 3000 & 50 & 8 & 2900 & 67 & 20 \\
        AT-RT4  & 125 & 25 & 4750 & 8 & 2900 & 32 & 4 & 2900 & 52 & 8 & 2900 & 67 & 20 \\
        AT-RT5 & 125 & 25 & 4750 & 8 & 3000 & 35 & 5 & 3000 & 50 & 10 & 3000 & 65 & 22\\ 
        \midrule
        & GT & VT \\
CC-RT0 & 0.5 & 3\\
        CC-RT1 & 0.7 & 3\\
        CC-RT2 & 0.7 & 4\\
        CC-RT3 & 0.7 &5\\
        \midrule
                &IL  &OL\\
        OMM-RT0 & 0  & 0\\
        OMM-RT1 & 0  & -0.5\\
        OMM-RT2 & -0.5  & 0\\
    \bottomrule
    \end{tabular}
    \label{tab:CCreqs}
\end{table*}

We considered six of the last versions of the CC model described in \Cref{tab:ATmodels}  (middle part).
According to the authors, versions CC-v6.1, CC-v6.2, and CC-v7.1 contain failures, while versions CC-v7.3, CC-v7.4, and CC-v7.5 do not.

We manually designed the RT for the requirements (CC-F1, CC-D1, CC-D2, CC-D3) of CC since they are specified using Test Assessment blocks from \SLTest~\cite{TestAssessment}.
\Cref{tab:requirements} (middle part) provides a textual description for each requirement.
We generated an RT modeling these four requirements.
As for the AT model, we parameterized each requirement. 
For example, CC-F1 can be instantiated by considering different values for the velocity tolerance (\textbf{VT}).
We then considered four instances (CC-RT0, CC-RT1, CC-RT2, CC-RT3) of this RT reported in \Cref{tab:CCreqs} (middle part) generated by assigning values to its parameters.

Therefore, we considered $24$ ($6\times4$) model-RT combinations obtained by combining each model and RT.
The original model-RT combinations are $\langle$CC-v6.1, CC-RT0$\rangle$, $\langle$CC-v6.2, CC-RT0$\rangle$, $\langle$CC-v7.1, CC-RT0$\rangle$, $\langle$CC-v7.3, CC-RT0$\rangle$, $\langle$CC-v7.4, CC-RT0$\rangle$, and $\langle$CC-v7.5, CC-RT0$\rangle$.

\emph{Observer Mode Model (OMM)}. We considered $12$ model-RT combinations from a tutorial example provided by MathWorks~\cite{RequirementsTable}.
We use the RT that is observing the model behavior without executing actions that consists of two requirements (OMM1 and OMM2).
The requirements from the RT are satisfied by the model.

We considered four versions (OMM-v0,$\ldots$,OMM-v3) of the OMM model: the original model and three versions obtained by applying changes detailed in \Cref{tab:ATmodels} (bottom part) to its parameters. 

We consider the RT provided by the example and parameterize each requirement of the RT: We added a threshold value for the minimum values the variables y1 and y2 can assume.
We considered three instances of this RT (OMM-RT0, OMM-RT1, and OMM-RT2) detailed in \Cref{tab:CCreqs} (bottom part)  generated by assigning values to its parameters.

Therefore,  we considered  $12$ ($4\times3$) model-RT combinations in total.

\emph{Tool Configuration}. The tool configurations for AT and OMM correspond to the ones used by ARIsTEO \cite{Aristeo} and ATheNA \cite{formica2022search} in the last edition of the ARCH competition. 
\changeRead{The time budget was set to 1500 iterations, with Simulated Annealing \cite{Abbas2014} as the search algorithm, and each run of the experiment was repeated $10$ times.}{search_methodAnswerb}{search_method}
On the other hand, the CC model has a much higher running time, so performing a single run with a limit of 1500 iterations could take more than 20 hours.
For this reason, we used the same testing condition as in \cite{10.1145/3611643.3613894} and use 20 iterations as the time budget.
The remaining parameters of the search algorithm are set to the default values used by S-Taliro (and HECATE).
\Cref{tab:staliroparam} lists all the parameters for Simulated Annealing and their values.

\begin{table}[t]
    \centering
    \caption{Values assigned to the configuration parameters of the search algorithms used by Hecate.}
    \label{tab:staliroparam}
    \begin{tabular}{p{5.5cm} p{2.5cm}}
        \toprule
        \textbf{Parameter}   &\textbf{Value}\\
        \midrule
        Optimization solver.    &SA\\
        Number of runs.         &$10$\\
        Max. iterations per run.       &$1500$ (AT, OMM)\\
                                                    &$20$ (CC)\\
Gradient Descent+SA    &False\\
Maximum Acceptance Ratio.        &$0.55$\\
        Minimum Acceptance Ratio.        &$0.45$\\
        Initial continuous beta parameter.   &$-15$\\
        Continuous beta adaptation parameter.   &$50\%$\\
        Initial discrete beta parameter.    &$-15$\\
        Discrete beta adaptation parameter. &$50\%$\\
        Initial displacement ratio for max step size.   &$0.75$\\
        Displacement adaptation parameter.  &$10\%$\\
        Upper Bound on displacement ratio. &$0.99$\\
        Lower Bound on displacement ratio. &$0.01$\\
Cost function evaluations before updating acceptance criteria.   &$50$\\
\bottomrule
    \end{tabular}
\end{table}

To run our SBST framework we had to design Parameterized Test Sequences for the AT and OMM models.
The CC model already has its own Parameterized Test Sequences, so we reuse them.
The Test Sequences for AT and OMM are designed by considering the original shape and range of the input signals.
For AT, we considered the range and function shape of Instance 2 of the ARCH competition.
Therefore, the two input signals are modeled as piecewise constant signals with a single discontinuity.
The inputs for OMM were originally constant signals, so we used a piecewise constant signal with a single discontinuity.
In total, the Parameterized Test Sequences used respectively five search parameters for AT and five for OMM.
The number of CC parameters is four for versions CC-v6.1, CC-v6.2, CC-v7.1, and CC-v7.3, six for version CC-v7.4, and nine for version CC-v7.5.
\Cref{tab:HecateParam} lists the names and ranges of our parameters.
Our SBST framework searches for values of these parameters that cause a violation of the requirements of the RTs, i.e., they return a failure-revealing test case.

\begin{table}
    \centering
    \caption{Search parameters names and ranges for the Test Sequences of the AT, CC, and OMM models.}
    \label{tab:HecateParam}
    \begin{tabular}{r l l}
        \toprule
        \textbf{Model}  &\textbf{Name}  &\textbf{Range}\\
        \midrule
        \multirow{5}{*}{AT}   &Hecate\_throttle1   &$[5,100]\%$\\
            &Hecate\_brake1         &$[0,325]\mathit{lb}\cdot \mathit{ft}$\\
            &Hecate\_throttle2      &$[5,100]\%$\\
            &Hecate\_brake2         &$[0,325]\mathit{lb}\cdot \mathit{ft}$\\
            &Hecate\_trans          &$[0,35]s$\\
        \midrule
        \multirow{9}{*}{CC}     &Hecate\_Transition1   &$[30,40]s$\\
            &Hecate\_Transition2$^{\ast,\dagger}$   &$[0,20]s$\\
            &Hecate\_desVel1        &$[120,150]km/h$\\
            &Hecate\_desVel2        &$[120,150]km/h$\\
            &Hecate\_desVel3$^{\ast,\dagger}$        &$[120,150]km/h$\\
            &Hecate\_slope          &$[-4,4]^{\circ}$\\
            &Hecate\_verShift$^{\dagger}$       &$[-1,1]^{\circ}$\\
            &Hecate\_period$^{\dagger}$         &$[30,200]s$\\
            &Hecate\_horShift$^{\dagger}$       &$[0,\pi]$\\
        \midrule
        \multirow{5}{*}{OMM}    &Hecate\_u1   &$[-2,5]$\\
            &Hecate\_u2             &$[-2,5]$\\
            &Hecate\_u3             &$[-2,5]$\\
            &Hecate\_u4             &$[-2,5]$\\
            &Hecate\_time           &$[2,7]s$\\
        \bottomrule
    \end{tabular}\\
    $^{\ast}$ Parameters added for model version CC-v7.4 \\ 
        $^{\dagger}$ Parameters added for model version  CC-v7.5.
\end{table}

\subsection{RQ1 --- Effectiveness}
\label{sec:effectiveness}
To assess the effectiveness of our SBST framework, we analyze its capability of generating failure-revealing test cases. 
We consider each of the \modelrtcombinations model-RT combinations from our benchmark, run our SBST tool, verify whether our tool returns a failure-revealing test case, and reflect on our results. 
Note that since no tool supports SBST of RTs, we do not have any baseline for our assessment.

Recall that each of our experiment run was repeated $10$ times.
\Cref{tab:RQ1results} analyzes the requirements violated by the failure-revealing test cases. 
The table is split horizontally into three parts containing the results for the AT, CC, and OMM examples.
The rows of the table report the versions of the RT. 
The columns are the versions of the model.
The column representing each version of the model is split into two columns.
We recall that each RT has several requirements (see \Cref{sec:benchmark}).
The columns labeled with ``F'', and ``Req'' respectively report if at least one of the requirements of the RT is violated (\cmark) or not (\xmark) in one of the $10$ runs, and the list of the requirements that are violated in at least one of the runs by the failure-revealing test cases.

In the following, we discuss the results obtained for each model.

\subsubsection{AT Model}
\Cref{tab:RQ1results} shows that our SBST framework could generate at least a failure-revealing test case for $\approx 96\%$ ($23$ out of $24$) of the model-RT combinations.
The tool could not find a failure-revealing test case for $\langle$AT-v3, AT-RT3$\rangle$ in any of the runs.
For this model-RT combination, our solution tried to violate requirement AT2 and failed due to the lowered torque capacity of model version AT-v3, i.e.,  the optimization algorithm entered a local minimum and did not reach the global minimum that was showing the violation of the requirement AT6c (as shown in \Cref{tab:RQ1resultsUR}).

In total, we found 219 failure-revealing test cases out of $240$ runs ($24$ model-RT combinations $\times 10$ runs).
We ran all the failure-revealing test cases through the appropriate Requirements Table and confirmed that each of these test cases violated at least one of the requirements of the RT.

The requirements AT1, AT2, AT6a, AT6b, AT6c are respectively violated in
$\approx 17\%$ ($38$ out of $219$),
$\approx 49\%$ ($107$ out of $219$),
$\approx 24\%$ ($53$ out of $219$),
$\approx  9\%$ ($20$ out of $219$), and 
$\approx 27\%$ ($59$ out of $219$)
of the failure-revealing runs.
Note that a failure-revealing test case can simultaneously violate more than one requirement.
Therefore, the sum of all the requirement violations is higher than $100\%$ (or $219$ failure-revealing runs).
The requirement AT2 is the one that is violated the most ($\approx 49\%$ of the times).
This result confirms that the proposed approach could expose diverse failures and is consistent with the results reported in the ARCH 2024 report: For the requirements encoded in our RT, this requirement requires the least number of iterations for being violated for most of the tools participating in the competition (seven out of eight).

\begin{table*}[t]
    \centering
    \caption{Results of our SBST framework for the AT, CC, and OMM models. For each model and RT version, it is reported if at least one of the requirements is violated (\cmark) or not (\xmark),
    and all the requirements that are violated.}
    \footnotesize
    \begin{tabular}{l | r p{1.45cm} | r  p{1.45cm} | r  p{1.45cm} | r  p{1.45cm} | r  p{1.45cm} | r  p{1.45cm} }
        \toprule
        \multirow{1}{1.3cm}{\textbf{RT}}    &$F$
        &$\mathit{Req}$ &$F$    &$\mathit{Req}$ &$F$    &$\mathit{Req}$ &$F$    &$\mathit{Req}$ &$F$    &$\mathit{Req}$ &$F$    &$\mathit{Req}$\\
        \midrule
        &\multicolumn{2}{c|}{AT-v0}  &\multicolumn{2}{c|}{AT-v1}  &\multicolumn{2}{c|}{AT-v2}  &\multicolumn{2}{c|}{AT-v3}\\
        \midrule
        AT-RT0      & \cmark    & AT2, AT6a, AT6b, AT6c         & \cmark    & AT2, AT6a, AT6b, AT6c         & \cmark    & AT1, AT2, AT6c        & \cmark    & AT2 \\
        AT-RT1      & \cmark    & AT1, AT2, AT6a, AT6b, AT6c    & \cmark    & AT1, AT2, AT6a, AT6b, AT6c    & \cmark    & AT1, AT2, AT6a        & \cmark    & AT2 \\
        AT-RT2      & \cmark    & AT1, AT6a, AT6b, AT6c         & \cmark    & AT1, AT2, AT6a, AT6c          & \cmark    & AT1, AT6a, AT6b, AT6c & \cmark    & AT6c \\
        AT-RT3      & \cmark    & AT1, AT6a, AT6b, AT6c         & \cmark    & AT1, AT2, AT6a, AT6b, AT6c    & \cmark    & AT1, AT6a, AT6b, AT6c & \xmark    & -- \\
        AT-RT4      & \cmark    & AT2, AT6a, AT6c               & \cmark    & AT1, AT2, AT6a, AT6c          & \cmark    & AT1, AT2, AT6a        & \cmark    & AT6a, AT6c \\
        AT-RT5      & \cmark    & AT1, AT2, AT6a, AT6b, AT6c    & \cmark    & AT1, AT2, AT6a, AT6b, AT6c    & \cmark    & AT1, AT2, AT6a        & \cmark    & AT6a, AT6b, AT6c \\
        \midrule
        &\multicolumn{2}{c|}{CC-v6.1}  &\multicolumn{2}{c|}{CC-v6.2}  &\multicolumn{2}{c|}{CC-v7.1}  &\multicolumn{2}{c|}{CC-v7.3} &\multicolumn{2}{c|}{CC-v7.4}  &\multicolumn{2}{c}{CC-v7.5}\\
        \midrule
        CC-RT0      & \cmark    & F1    & \cmark    & F1    & \cmark    & F1    & \cmark    & F1    & \cmark    & F1    & \cmark    & F1 \\
        CC-RT1      & \cmark    & F1    & \cmark    & F1    & \cmark    & F1    & \cmark    & F1    & \cmark    & F1    & \cmark    & F1 \\
        CC-RT2      & \cmark    & F1    & \cmark    & F1    & \cmark    & F1    & \cmark    & F1    & \cmark    & F1    & \cmark    & F1 \\
        CC-RT3      & \xmark    & --    & \xmark    & --    & \xmark    & --    & \cmark    & F1    & \cmark    & F1    & \cmark    & F1 \\
        \midrule
        &\multicolumn{2}{c|}{OMM-v0}  &\multicolumn{2}{c|}{OMM-v1}  &\multicolumn{2}{c|}{OMM-v2}  &\multicolumn{2}{c|}{OMM-v3}\\
        \midrule
        OMM-RT0     & \xmark    & --            & \cmark    & OMM2          & \cmark    & OMM1, OMM2    & \cmark    & OMM1, OMM2 \\
        OMM-RT1     & \xmark    & --            & \xmark    & --            & \xmark    & --            & \xmark    & -- \\
        OMM-RT2     & \cmark    & OMM1, OMM2    & \cmark    & OMM1, OMM2    & \cmark    & OMM1, OMM2    & \cmark    & OMM1, OMM2 \\
        \bottomrule
    \end{tabular}
    \label{tab:RQ1results}
\end{table*}

\begin{table*}[t]
    \centering
    \caption{Failure Rate (FR), Average ($\bar{S}$) and Median ($\widetilde{S}$) number of iterations needed to find the failure-revealing test case for each model-RT combination.}
    \footnotesize
    \begin{tabular}{l | r r c | r r c | r r c | r r c | r r c | r r c }
        \toprule
        \multirow{1}{1.3cm}{\textbf{RT}}    &\textbf{$FR$}  &\textbf{$\bar{S}$} &\textbf{$\widetilde{S}$}   &\textbf{$FR$}  &\textbf{$\bar{S}$}        &\textbf{$\widetilde{S}$}   &\textbf{$FR$}   &\textbf{$\bar{S}$}     &\textbf{$\widetilde{S}$}   &\textbf{$FR$}      &\textbf{$\bar{S}$}    &\textbf{$\widetilde{S}$}   &\textbf{$FR$}   &\textbf{$\bar{S}$}     &\textbf{$\widetilde{S}$}   &\textbf{$FR$}     &\textbf{$\bar{S}$}    &\textbf{$\widetilde{S}$} \\
        \midrule
        &\multicolumn{3}{c|}{AT-v0}  &\multicolumn{3}{c|}{AT-v1}  &\multicolumn{3}{c|}{AT-v2}  &\multicolumn{3}{c|}{AT-v3}\\
        \midrule
        AT-RT0      & 10    & 16.2  & 11    & 10    & 15.7  & 13    & 10    & 14.7  & 14    & 10    & 34.1  & 24  \\
        AT-RT1      & 10    & 18.2  & 13    & 10    & 10.1  & 6     & 10    & 7.4   & 6     & 10    & 34.1  & 24  \\
        AT-RT2      & 10    & 82.5  & 48    & 10    & 17.5  & 11    & 10    & 31.6  & 10    & 8     & 611.4 & 558 \\
        AT-RT3      & 10    & 112.0 & 81    & 10    & 22.9  & 18    & 10    & 35.8  & 33    & 0     & --    & --  \\
        AT-RT4      & 10    & 16.4  & 11    & 10    & 23.8  & 22    & 10    & 11.4  & 8     & 2     & 491.0 & 491 \\
        AT-RT5      & 10    & 16.2  & 9     & 10    & 29.2  & 20    & 10    & 11.4  & 8     & 9     & 509.8 & 443 \\
        \midrule
        &\multicolumn{3}{c|}{CC-v6.1}  &\multicolumn{3}{c|}{CC-v6.2}  &\multicolumn{3}{c|}{CC-v7.1}  &\multicolumn{3}{c|}{CC-v7.3} &\multicolumn{3}{c|}{CC-v7.4}  &\multicolumn{3}{c}{CC-v7.5}\\
         \midrule
        CC-RT0      & 5     & 5.6   & 4     & 5     & 5.6   & 4     & 6     & 9.7   & 10    & 10    & 4.4   & 2     & 9     & 5.0   & 4     & 9     & 7.0   & 8 \\
        CC-RT1      & 5     & 5.6   & 4     & 5     & 5.6   & 4     & 6     & 9.7   & 10    & 10    & 4.3   & 1     & 9     & 5.3   & 4     & 9     & 7.1   & 8 \\
        CC-RT2      & 2     & 14.0  & 14    & 2     & 14.0  & 14    & 2     & 7.5   & 7     & 7     & 3.7   & 1     & 9     & 4.9   & 4     & 5     & 11.8  & 12 \\
        CC-RT3      & 0     & --    & --    & 0     & --    & --    & 0     & --    & --    & 5     & 4.6   & 5     & 3     & 9.0   & 11    & 3     & 14.3  & 17 \\
        \midrule
        &\multicolumn{3}{c|}{OMM-v0}  &\multicolumn{3}{c|}{OMM-v1}  &\multicolumn{3}{c|}{OMM-v2}  &\multicolumn{3}{c|}{OMM-v3}\\
         \midrule
        OMM-RT0     & 0     & --    & --    & 10    & 68.5  & 41    & 10    & 35.9  & 29    & 10    & 34.4   & 12 \\
        OMM-RT1     & 0     & --    & --    & 0     & --    & --    & 0     & --    & --    & 0     & --     & -- \\
        OMM-RT2     & 10    & 7.6   & 7     & 10    & 7.6   & 7     & 10    & 7.6   & 7     & 10    & 7.0    & 5 \\
        \bottomrule
    \end{tabular}
    \label{tab:RQ2results}
\end{table*}

\begin{figure*}
    \centering
    \hfill
    \begin{subfigure}{0.25\textwidth}
        \includegraphics[width=\textwidth]{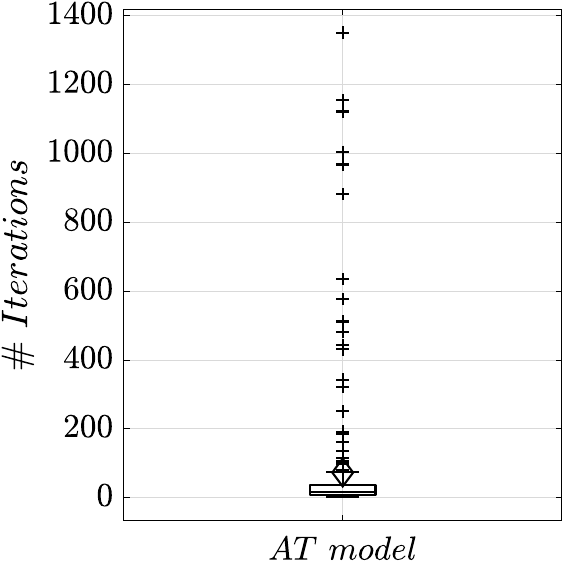}
\end{subfigure}
    \hfill
    \begin{subfigure}{0.25\textwidth}
        \includegraphics[width=\textwidth]{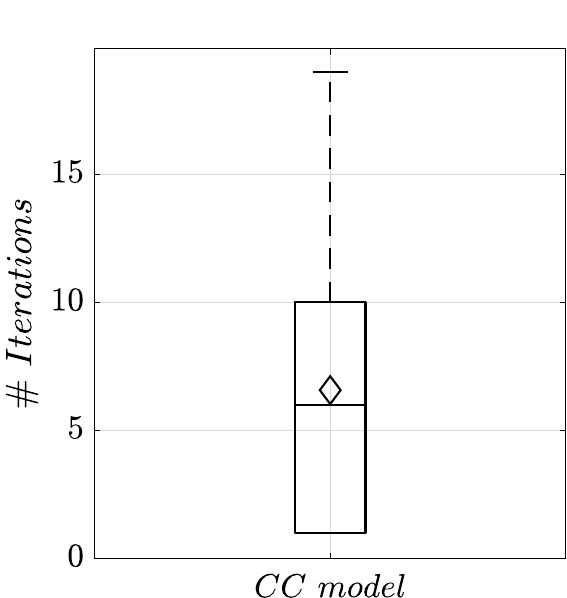}
\end{subfigure}
    \hfill
    \begin{subfigure}{0.25\textwidth}
        \includegraphics[width=\textwidth]{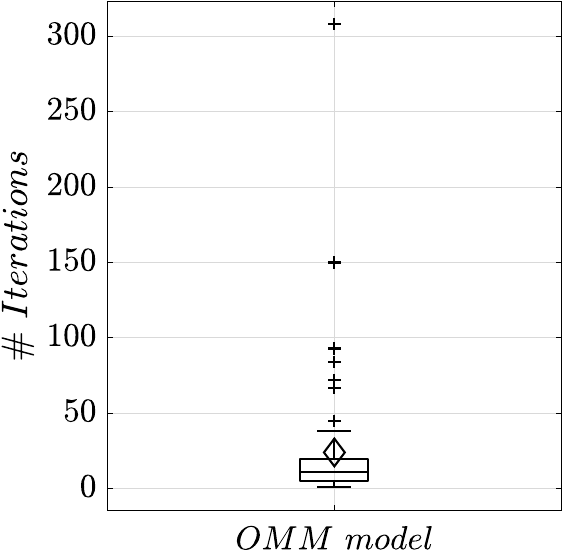}
\end{subfigure}
    \hfill
    \caption{Number of iterations for the failure-revealing runs. Diamond depicts the average.}
    \label{fig:boxplots}
\end{figure*}

\subsubsection{CC Model}
Our SBST framework could generate at least a failure-revealing test case for $\approx 88\%$ ($21$ out of $24$) of the model-RT combinations.
All the failure-revealing test cases were verified using the appropriate RT checking that at least one of the requirements of the RT was violated.

In total, we found 126 failure-revealing test cases out of $240$ runs ($24$ model-RT combinations $\times 10$ runs).
The combinations where our solution could not produce a failure-revealing test case in 10 runs were $\langle$CC-v6.1, CC-RT3$\rangle$, $\langle$CC-v6.2, CC-RT3$\rangle$, and $\langle$CC-v7.1, CC-RT3$\rangle$.
For version CC-RT3 of the RT, we increase the grace period (GT) and the vehicle speed (VT) thresholds. This change likely makes the requirements impossible to be violated in the model versions CC-v6.1, CC-v6.2, and CC-v7.1.

Considering the requirements (RT0) and models (CC-v6.1, CC-v6.2, CC-v7.1, CC-v7.3, CC-v7.4, CC-v7.5) reported in the original publication~\cite{10.1145/3611643.3613894}, our tool could return a failure-revealing test cases for all the model-RT combinations.
\emph{Surprisingly these model-RT combinations include} $\langle$CC-v7.3, CC-RT0$\rangle$, $\langle$CC-v7.4, CC-RT0$\rangle$, \emph{and} $\langle$CC-v7.5, CC-RT0$\rangle$ \emph{for which the existing tool could not previously return any failure-revealing test case~\cite{10.1145/3611643.3613894}}. 
We tried one test case of these three model-RT combinations by performing a Hardware-in-the-Loop experiment with an industrial simulator.
Videos about these testing scenarios are available in the Replication package.
These experiments confirm the findings of the tool: While the drivability requirements are satisfied, the functional requirement is violated in all the experiments.

The requirements CC-F1, CC-D1, CC-D2, and CC-D3 are, respectively, violated in $100\%$ (126 out of 126), $0\%$ (0 out of 126), $0\%$ (0 out of 126), and $0\%$ (0 out of 126) of the failure-revealing runs.
Unlike the original publication~\cite{10.1145/3611643.3613894}, which showed the violation of both the functional and drivability requirements, for all our cases, when the requirement was violated, our SBST framework reported that the violation was in the functional requirement CC-F1.
Unfortunately, setting up VI-Grade~\cite{VIgrade}, the driving simulator used by the VI-CarRealTime simulator, on a new laptop required recreating the vehicle model, which may consider slightly different physical characteristics (e.g., its mass) and may lead to a discrepancy in the test results.

\subsubsection{OMM model} 
Our SBST framework generated at least a failure-revealing test case for $\approx 58\%$ ($7$ out of $12$) of the model-RT combinations.
We checked the failure-revealing test cases against the appropriate RT, which confirmed that each of these test cases violated at least one of the requirements of the RT.

In total, we found 70 failure-revealing test cases out of $120$ runs ($12$ model-RT combinations $\times 10$ runs).
For the requirements combinations $\langle$OMM-v0, OMM-RT0$\rangle$, $\langle$OMM-v0, OMM-RT1$\rangle$, $\langle$OMM-v1, OMM-RT1$\rangle$, $\langle$OMM-v2, OMM-RT1$\rangle$, and $\langle$OMM-v3, OMM-RT1$\rangle$ our solution did not reveal any violation of the requirements of the RT.
This is not surprising since for these combinations, if the input signals satisfy the preconditions of the requirements, then the output signals will always satisfy the postconditions.
For example, for the combination $\langle$OMM-v3, OMM-RT1$\rangle$ the model ensures that if the first input signal is greater than $0$ (precondition of OMM1), then the first output signal will always be in the range $(-0.5,10.2]$.
This ensures the requirement OMM1 is always satisfied as the first output signal can't be lower than the threshold value of $-0.5$ and is always lower than $10.2$ (postcondition of OMM1), making the violation of the requirement impossible.

The requirements OMM1 and OMM2 are respectively violated in
$\approx 27\%$ ($19$ out of $70$), and 
$\approx 73\%$ ($51$ out of $70$)
of the failure-revealing runs.
This result confirms that the proposed approach could expose diverse failures.

\subsubsection{Summary and Answer}
In summary, our SBST framework returned at least a failure-revealing test case for \rqonetotal (51=23+21+7 out of \modelrtcombinations) of our model-RT combinations.
This result is comparable to the results reported by the last edition of the ARCH competition:
83\% for ARIsTEO,
88\% for ATheNA,
83\% for EXAM-Net,
75\% for FalCAuN~\cite{Waga20},
100\% for ForeSee~\cite{falsQBRobCAV2021},
100\% for FReaK~\cite{FReaK},
33\% for Moolight~\cite{MoonlightTool},
100\% for NNFal~\cite{kundu2024data},
79\% for OD~\cite{STGEM},
and 88\% for $\Psi$-Taliro~\cite{psytalirotool}.

\begin{Assumption}[RQ1 --- Answer]
Our SBST framework returned a failure-revealing test case for \rqonetotal of the model-RT combinations.
Remarkably, three model-RT combinations refer to a CC example that other tools already tested without finding any requirement violation.
\end{Assumption}

\subsection{RQ2 --- Efficiency}
\label{sec:efficiency}
To assess the efficiency of our solution we recorded the number of iterations required to generate the failure-revealing test cases since (a)~this is the main metric used in the ARCH competition, and (b) each iteration of our tool takes constant time (with small perturbations caused by the testing platform). 
For example, one iteration of our solution for the AT, CC, and OMM model takes respectively $2.8s$, $46.7s$, and $1.1s$.
Therefore, the total execution time is (approximately) the product between the number of iterations and the execution time of one iteration.

\Cref{tab:RQ2results} presents our results.
The table is split horizontally into three parts containing the results for the AT, CC, and OMM examples.
For each model-RT combination, the table reports the falsification rate (FR),  average ($\bar{S}$), and median ($\tilde{S}$) number of iterations across the $10$ runs.
Empty cells indicate that none of the runs could detect a failure-revealing test case for that specific RT.

\subsubsection{AT Model} 
For the $219$ failure-revealing test case produced by our SBST, in $89\%$ ($194$ out of $219$) of the cases it was computed in less than 100 iterations.
For $10\%$ ($21$ out of $219$) of the cases, the failure-revealing test cases were computed in over 100 iterations but less than 1000.
For $2\%$ ($4$ out of $219$) of the cases, the failure-revealing test cases were computed in more than 1000 iterations.
The model-RT combinations that used more than 1000 iterations are $\langle$AT-v3, AT-RT2$\rangle$ (in 2 runs out of 10) and $\langle$AT-v3, AT-RT5$\rangle$ (in 2 runs out of 10).
This is reasonable because the model version AT-v3 can exert less torque than the original version, which makes it harder to violate requirements with a vehicle speed threshold.
Furthermore, AT-RT2 and AT-RT5 respectively increase the threshold value for AT2 and reduce the time horizon for AT2, which is the requirement that was violated the most in the other combinations.
Our solution requires the highest number of iterations (avg=285.7) for AT-v3 and the lowest for AT-v1 (avg=19.9) and AT-v2 (avg=18.7).
This result is expected: The change applied to AT-v3  lowered the torque on the transmission making it harder to break the speed limits, while the changes applied to AT-v1 and AT-v2 made it easier.

Comparing the result of the combination $\langle$AT-v0, AT-RT0$\rangle$, which is made from the original model and requirement, with the results reported in the ARCH2024 report, the average number of iterations required by our solution ($16.2$) is lower than the one from ATheNA ($62.3$ for AT2), EXAM-Net ($22.7$ for AT2), FalCAuN ($256$ for AT2), ForeSee ($115.2$ for AT6a), and OD ($18.5$ for AT2) but higher than ARIsTEO ($15.1$ for AT2), FReaK ($2.2$ for AT2), NNFal ($1.5$ for AT1), and $\Psi$-Taliro ($13.7$ for AT2). 
Therefore, the number of iterations is higher than four of the competitors and lower than five.
Compared with tools that rely on the same search technology (ATheNA, ARIsTEO, and $\Psi$-Taliro),
our tool was faster than ATheNA requiring $46.1$ iterations less, and required only $1.1$ iterations more than ARIsTEO and $2.5$ more than Psy-Taliro.
We conjecture that this performance improvement depends on the number of search parameters which reduces the size of the search space: Our test sequences rely on 5 parameters (see \Cref{sec:benchmark}), while ATheNA, ARIsTEO, and S-Taliro use 10 control points.

\subsubsection{CC Model} 
For the $126$ failure-revealing test case produced by our SBST, in $76\%$ ($96$ out of $126$) of the cases, it was computed in less than 10 iterations. 
For the $24\%$ ($30$ out of $126$) of the test cases, the SBST used more than 10 iterations and less than 20 (maximum number of iterations).

Comparing the results related to the model-RT combinations ($\langle$CC-v6.1, CC-RT0$\rangle$, $\langle$CC-v6.2, CC-RT0$\rangle$, $\langle$CC-v7.1, CC-RT0$\rangle$) reported in the original CC paper, our solution is slightly more efficient: It required on average respectively $2.4$ and $4.3$ iterations less for $\langle$CC-v6.1, CC-RT0$\rangle$ and $\langle$CC-v7.1, CC-RT0$\rangle$ and $3.6$ iterations more for $\langle$CC-v6.2, CC-RT0$\rangle$.
However, these differences are negligible in practice since running $4$ iterations requires approximately $180$ seconds.

\subsubsection{OMM Model}
For the $70$ failure-revealing test case produced by our SBST, in $97\%$ ($68$ out of $70$) of the cases, it was computed in less than 100 iterations.
For $3\%$ ($2$ out of $70$) of the cases, the failure-revealing test cases were computed in over 100 iterations but less than 1000.

We could not compare our results with previous work: These combinations were never considered by any other SBST framework.

\subsubsection{Summary and Answer}
The boxplots in \Cref{fig:boxplots} represent the number of iterations required to generate the failure-revealing test cases for each model, for all its model-RT combinations and failure-revealing runs.
Our SBST solution required  on average 
73.4 (\textit{min}=1,  \textit{max}=1351,
\textit{StdDev}=195.5), 
6.6 (\textit{min}=1,    \textit{max}=19,
\textit{StdDev}=5.5), 
24.1 (\textit{min}=1,   \textit{max}=308,
\textit{StdDev}=43.6), iterations respectively for the AT, CC, OMM.
The average is lower than 100 iterations for all of our models.
These results confirm the efficiency of our tool.
Note that a high-standard deviation is expected: \Cref{tab:RQ2results} shows that, for each model, the complexity of identifying failure-revealing test cases changes across the different versions of the model (e.g., AT-v0, AT-v1, AT-v2, and AT-v3 for AT) and RT (e.g., AT-RT0, AT-RT1, AT-RT2, AT-RT3, AT-RT4, and AT-RT5 for AT).

\begin{Assumption}[RQ2 --- Answer]
On average, our SBST framework required $73.4$, $6.6,$ and $24.1$ iterations respectively for the AT, CC, and OMM model.
For model-requirement combinations that were analyzed by other tools from the literature, existing tools show comparable performance.
\end{Assumption}

 \subsection{Reflections and Threats to Validity}
\label{sec:discussion}
\Cref{sec:relevance} discusses the relevance of our contribution and its industrial applicability.
\Cref{sec:threats} presents threats to validity.
\Cref{sec:comparison} reflects the comparison of our solution with existing ones from the literature.

\subsubsection{Relevance of the Contribution}
\label{sec:relevance}
Developing effective SBST solutions for CPS is a relevant software engineering problem~\cite{papadakis2019mutation,5210118}. 
Providing support for tabular expressions (e.g., those specified as RT) is \emph{important} since they have been used in industrial practice~\cite{crow1998formalizing,10.1145/2038642.2038676,menghi2024completeness}.
Therefore, our approach will significantly impact software engineering practices by enabling the automatic generation of test cases from RT.
Note that RTs enable engineers to use Boolean expressions to specify their pre- and post-conditions.
Boolean logic is significantly less expressive than temporal logic. Unlike Boolean logic, temporal logic formulas can nest temporal operators, making the formula more complex.
Our contribution is a test-case generation framework that directly supports RTs. 
This contribution is relevant because RTs are already part of the Simulink framework. Therefore, our solution will not require engineers to learn any new formal language.

Our approach is driven by translating the RT into a set of parallel Finite State Machines that perform runtime monitoring of requirements and robustness computation.
This translation is explicitly designed to support the testing activity and is integrated with existing search techniques. 
Our solution relies on well-established search-based testing techniques, a recommended software engineering practice.
Unlike other existing techniques (e.g., \cite{Panichella_2018}) our technique is driven by the system requirements expressed as RT and defined in the context of Simulink.
Other solutions are not tailored, designed, or applicable to this problem. 
In summary, our contribution is \emph{original} since we are the first to propose a black-box SBST approach for RT.

Other work from the literature inspired our solution.
A recent work that combines multiple requirements of a CPS into a single conjunctive requirement~\cite{peltomaki2022falsification} inspired our interpretation of RT with multiple requirements (\Cref{def:multiple}).
The robustness semantics of STL~\cite{FainekosPappas2009}, the quantitative semantics of RFOL~\cite{menghi2019generating}, and multi-valued semantics of logic formalisms (e.g.,~\cite{larsen1988modal}) inspired the replacement of the logical operators of the Boolean expression using $min$ and $max$ operators.
The approaches for generating automated and online test oracles for Simulink models~\cite{menghi2019generating} and supporting SBST of Test Sequence and Test Assessment Blocks~\cite{Hecate2024} inspired our aggregator component.
\changeRead{To summarize, instead of developing a new solution from scratch, we reused and adapted these solutions to solve our problem, a recommended software engineering practice.
Instead of Stateflow, we could have used other target languages (e.g., logic) for our translation.
We decided to use  Stateflow since it is the standard tool for implementing state machines in Simulink.
Our evaluation confirms the effectiveness of this decision.}{other_transformationsAnswer}{other_transformations}
Finally, we implemented our solution as a plugin for HECATE, which makes our solution integrated with existing tools and more easily testable and maintainable.

The objective of our technique is to find a failure-revealing test case. 
Therefore, it stops when at least one requirement is violated (see~\Cref{sec:hecate}).
Technically, our tool can be configured to prevent our solution from stopping at the first requirement violation and terminate its execution when the maximum number of iterations is used.
This can enable the tool to find violations of other requirements.

\changeRead{
Like other works, which lifted the robustness semantics of timed automata~\cite{gupta1997robust,henzinger2000robust} to temporal logics (e.g.,~\cite{FainekosPappas2009,fainekos2009robustness}) and aspect-oriented design (e.g.,~\cite{ali2012modeling,ali2013assessing}), we lifted the robustness semantics to RT.
Like other works that translated formal specifications into test oracles (e.g.,~\cite{9108630,DBLP:journals/tse/PadghamZTM13,DBLP:conf/icse/HeCHWPY19,DBLP:conf/icse/MotwaniB19,menghi2019generating}), we translated RT into test oracles.
The technical challenges addressed by our work are comparable to existing works from the research literature, despite being unique: RT is significantly different than other formalisms in terms of language constructs and expressivity.}{small_contribution_1Answer}{small_contribution_1}\change{}{small_contribution_2Answer}{small_contribution_2}

\change{There are several technical difficulties our solution faces. We summarized the technical difficulties for each of the contributions reported in \Cref{sec:intro}.
For the selection of the target tool (Stateflow Charts) to be used for our translation (Contribution~\textbf{C1}), we faced several challenges.
We had to find a formalism that could
(a)~express and encode the features provided by RT, since formalisms with limited expressiveness might not be able to encode all the requirements,
(b)~enable a ``clean'' and ``simple'' translation since this would enable proving its soundness easily,
(c)~check the requirements as the simulation is running, and 
(d)~be integrated within the \simulink platform with an easy installation process. 
This last point can be significant to foster the adoption of the approach in the industry, since many companies have strict guidelines on the approval of external software.
The selection of Stateflow Charts was the outcome of a long process in which we considered several alternative solutions and assessed their pros and cons. 
We selected Stateflow Charts since they satisfy all these desired criteria.
For the definition of the quantitative semantics of Requirements Tables (Contribution~\textbf{C2}),
it was the outcome of a meticulous and accurate analysis of the semantics of the RTs~\cite{menghi2024completeness}. During this process, we had to define a quantitative measure for the satisfaction of the different constructs of the RT and its requirements. 
Additionally, as argued in \Cref{sec:hecate}, we had to select an interpretation for the well-formedness properties when multiple requirements are associated with a quantitative measure. 
For the definition of the translation that converts RTs into fitness functions (Contribution~\textbf{C3}), we had to overcome significant challenges and corner cases, such as defining procedures to handle (a)~multiple requirements simultaneously, each with its precondition, postcondition, and possible duration, (b)~the \prev and the \duration operators, which are not available in Stateflow Charts (\Cref{sec:prev}), (c)~prove the soundness of our translation. 
For the definition of the implementation of our solution (Contribution~\textbf{C4}), we had to integrate our approach within HECATE.
This required us to develop a script that programmatically reads RT and converts them in Stateflow Charts.
We encountered several challenges. 
For example, in the evaluation phase, we initially thought that the problem we were facing when deploying our experiments on Compute Canada was caused by our translation. 
Only after a careful debugging activity, we realized that the Cedar cluster
uses a Linux distribution incompatible with Simulink 2022a,
it is unable to openSimulink for 2022b,
2023a and 2023b contain a bug that prevents us from reading RTs input and output signals (faulty implementation of the function ``slreq.modeling.RequirementsTable/findSymbol''), and
2024b requires at least 48Gb of memory to execute the 10 runs (possibly due to a memory leak).
For the benchmark definition (Contribution~\textbf{C5}), we first analyzed existing research literature to verify the existence of any alternative benchmark that could be used in our evaluation, but found none.
Therefore, we considered existing benchmarks from the research literature, looking for models that contain multiple requirements that could be formalised using the requirements of the Requirements Tables.
We found three practical examples we could use for our evaluation: Two containing specifications in signal-temporal logic (STL), and one containing specifications expressed as Requirements Tables. 
Once we found these examples, we had to define our benchmark. We defined \modelrtcombinations model-RT combinations by formalizing different Requirements Tables for these models.
For the empirical evaluation (Contribution~\textbf{C6}), we had to identify a baseline to compare our results, define the configuration of the tools, and the experimental setup.
For each model, we critically analyzed our results, and for one of our models also performed hardware-in-the-loop testing.
Finally, we had to rigorously analyze our experimental set up and configuration settings to identify threats to validity, discuss the industrial applicability of our solution (Contribution~\textbf{C7}).
}{TechnicalAnswer}{tech_difficulties_stateflow} \change{}{tech_difficulties_stateflow_3}{small_contribution_3}
    
\delete{Running the experiments on the Compute Canada platform was hindered by some technical difficulties. 
Out of the five versions (2022a, 2022b, 2023a, 2023b, 2024b) of Matlab available on the platform: 
2022a is not fully compatible with the Linux distribution used on the Cedar cluster and cannot save the Simulink cache files,
2022b is not installed correctly on the Cedar cluster and is unable to open Simulink,
2023a and 2023b contain a bug that prevents us from reading Requirements Tables input and output signals (faulty implementation of the function ``slreq.modeling.RequirementsTable/findSymbol''),
and 2024b requires at least 48Gb of memory to execute the 10 runs (possibly due to a memory leak).}{RedundantParagraph}{tech_difficulties_stateflow}

Our contribution is \emph{sound}. Theorem~\ref{sec:soundTheorem} shows that our fitness function satisfies the properties from \Cref{def:ftprop}.
We also verified that the fitness value returned by our solution for our test cases matched the verdicts produced by the original RT: We implemented a script that checked that whenever the fitness value is positive all the requirements of the RT are satisfied, and if it is negative at least one requirement of the RT is violated. 
The script relies on the built-in functions of Matlab/Simulink to assess the satisfaction of the requirements.
These built-in functions do not provide a quantitative result, i.e., they do not provide a ``degree'' of satisfaction (robustness).
Our scripts did not reveal any problem in our solution.

We \emph{evaluated} our solution by analyzing its effectiveness and efficiency by critically examining our results. We compared the results of model-RT combinations from the literature with the results reported in the corresponding publication.
We executed 10 runs of our experiments as mandated by the ARCH competition~\cite{Arch2024}.
We shared our models, data, and tool supporting the transparency, reproducibility, and replicability of our solution.

\subsubsection{Threats to Validity}
\label{sec:threats}
Some experimental decisions threaten the external and internal validity of our results.

\emph{External Validity.} 
The selection of the model-RT combinations of our benchmark could threaten the external validity of our results.
We had to define our benchmark for the evaluation (an additional contribution made by this work) since there is no publicly available benchmark we could reuse and software companies typically do not share their Simulink models and RTs. 
\changeRead{Our benchmark is comprehensive: it includes models from different sources (ARCH Competition~\cite{ARCH14}, a recent publication~\cite{10.1145/3611643.3613894}, and the Mathworks website~\cite{RequirementsTable}).
Note that the system under tests \emph{were not manually created} and they have a clear purpose described in the corresponding publications.
The requirements are expressed using different formalisms (STL, Test Assessments, and RTs).
We converted the requirements expressed as the STL formulae and Test Assessments into RT.
We also tried to use other RTs from the  MathWorks documentation which we could not use: (a)~\cite{TutRT2, TutRT5, TutRT8, TutRT9} did not include any post-condition in the RT making the violation of the requirements not possible;
(b)~\cite{TutRT3, ReqTableDocumentationDuration, TutRT7} provided the RTs but no system model, and (c)~\cite{TutRT6} provided no information on the structure of the input signals, preventing the generation of test cases.
The number, size, and type of models considered in this work are consistent with similar works from this domain (e.g.,~\cite{Aristeo,Hecate2024})}{model_informationAnswer}{model_information}.

\emph{Internal Validity.} 
\changeRead{The values selected for the configuration parameters of our SBST solution could threaten the internal validity of our results. 
However, we mitigated this threat by considering the same values selected for the ARCH competition or in the original model publication.}{unfocused_evalAnswer}{unfocused_eval}

\subsubsection{Comparison with other Tools}
\label{sec:comparison}
We discuss tools that use different specification languages and search procedures.

\emph{Specification Languages}.
\changeRead{Our tool is designed to support requirements expressed as RT. RT are significantly different than other specifications tools. 
Unlike RT, STL can not specify the relation between the value of a signal in a time instant and its previous value since the atomic propositions of STL refer to a single time instant.
For example, STL can not force the speed to increase by 10 every second~\cite{menghi2024completeness}.
\change{In this sense, Tabular Requirements are more similar to the hybrid logic of signals (HLS - reference~\cite{menghi2021trace}), which enables the expression of these requirements.}{weak_motivationAnswerb}{weak_motivation}
Complementarily, STL can express specifications that involve complex temporal relations between events that RT can not express without introducing additional signals into the model.}{weak_motivationAnswer}{weak_motivation}
RTs are also significantly different than other formalisms (e.g., Test Assessments) since they are designed to support requirements analysis and design. 
\change{For example, unlike Test Assessment blocks, which define the behavior of a state machine that operatively describes a test oracle by defining which conditions must be checked and when, Requirements Tables provide a declarative description of the requirements of the system.
Technically, Test Assessment blocks are represented as state machines with transitions labeled with operators (e.g., after) that define how the state machine moves from one state to another. Based on this difference, our contribution is large and significant.
Considering a recent survey on the topic~\cite{barr2014oracle},
the Test Assessments can be classified as ``Specificed Test Oracles'' of type ``State Transition Systems'', while the oracle produced by this work are from a completely different category: They can be classified as ``Derived Test Oracles'', from ``Textual Documentation'' since Requirements Tables can be considered as a form of ``Restricting Natural Language''.}{small_contribution_1AnswerRTAns}{small_contribution_1}

In our evaluation, we considered requirements from the ARCH competition since a benchmark for SBST with RT does not exist and this competition provided us with a valuable benchmark of requirements.
We did not consider all the requirements and models available.
Specifically, we did not consider requirements with complex nested temporal operators (AT51, AT52, AT53, AT54, AFC27, NN, NNx, CC2, CC3, CC4, CC5) that would have required us to add additional signals to our model and therefore introduce biases in our evaluation.
We also did not consider models with a single requirement (F16, SC, PM) since they would produce Requirements Tables with a single row.
After removing these requirements only the AT model, with the requirements AT1, AT2, AT6a, AT6b, AT6c, AT6abc, the CC model, with requirements CC1, CCx, and the AFC model, with requirements AFC29 and AFC33 were still available for our evaluation.
AFC was excluded since AFC29 and AFC33 must be tested with different input ranges; therefore, they are equivalent to a model with a single requirement.
CC was excluded since the requirements CC1 and CCx have the same structure and lead to a simple fitness function. 

This paper does not argue against using formal specifications and Test Assessments and existing SBST frameworks that support these languages.
Our goal is not to compare the usability of RT and other modeling languages (e.g., STL or Test Assessments), and we are not claiming that users may prefer using RT rather than other specification languages. 
Our objective is to extend the support of existing tools (i.e., HECATE~\cite{Hecate2024}) to support requirements expressed in a tabular format (i.e., RT). 
Therefore, a comparison between our solution and these frameworks is unjustified. 

Simulink Design Verifier (SLDV) is the only tool that can generate test cases from an RT~\cite{TestCaseSDV}.
However, a comparison with this tool is not performed since it can only be used for the OMM model.
SLDV does not support some Simulink blocks used by the AT and CC models (e.g., the Integrator block).

\emph{Testing-strategies}. 
Our contribution is the translation of RT into a fitness function. 
This translation can drive different search-based procedures. 
\changeRead{Currently, our implementation supports extended ant colony optimization~\cite{AntColony}, genetic algorithms~\cite{GeneticAlgorithm}, uniform random~\cite{UniformRandom}, cross-entropy method~\cite{CrossEntropy}, and several other search methods which are reused from S-Taliro.
The implementation of these procedures is not a contribution to this work.}{search_methodAnswerC}{search_method}
Therefore, comparing different search-based strategies (e.g., Simulated Annealing --- SA --- with Uniform Random --- UR) is out of scope since they are not part of our contribution.
Furthermore, an extensive comparison of SBST techniques can not be based on the results obtained from three models and should rely on a larger set of models to have results that are statistically significant~\cite{Arcuri_2012}.
However, to show that our fitness function can effectively drive the search procedure, we performed a qualitative comparison between our solution configured with UR and SA algorithms. 
Note that our goal is to show that our fitness function can effectively drive the search procedure; It is not to show that SA is more (or less) effective than UR.

We consider the same model-RT combinations used in our evaluation and run the experiments from RQ1 and RQ2 using UR.
Table~\ref{tab:RQ1resultsUR} reports in how many out of the $10$ runs at least one of the requirements was violated (column ``FR''), and the requirement violated most frequently (column ``Req'') by UR.
Refer to \Cref{tab:RQ1results} and \Cref{tab:RQ2results} for the corresponding results using the SA search algorithm.
\Cref{fig:boxplotsUR} contains the number of iterations required for the three models to compute the failure-revealing runs for UR. For convenience, we also report the results obtained by SA in RQ2.
Surprisingly, for two models (AT and CC) over three, UR was more effective than SA.
It could find more failure-revealing test cases with fewer iterations.
From \Cref{sec:efficiency}, SA required  on average respectively
73.4 (\textit{min}=1,  \textit{max}=1351,
\textit{StdDev}=195.5), 
6.6 (\textit{min}=1,    \textit{max}=19,
\textit{StdDev}=5.5), and
24.1 (\textit{min}=1,   \textit{max}=308,
\textit{StdDev}=43.6) iterations for the AT, CC, and OMM models.
UR algorithm required on average respectively
25.9 (\textit{min}=1,  \textit{max}=453,
\textit{StdDev}=60.3), 
5.3 (\textit{min}=1,    \textit{max}=19,
\textit{StdDev}=4.7), and
39.0 (\textit{min}=1,   \textit{max}=268,
\textit{StdDev}=69.1) for the AT, CC, and OMM models outperforming (on average) SA by 47.5 iterations for the AT model and 1.3 iterations for the CC model.
SA outperformed UR by 14.9 iterations for the OMM model.

\begin{table*}[t]
    \centering
    \caption{Results of our SBST framework for the AT, CC, and OMM models using the \textbf{Uniform Random} (UR) algorithm. For each model and RT version, it is reported in how many runs (out of 10) at least one of the requirements is violated (FR), and which requirement is violated more frequently (Req).}
    \footnotesize
\begin{tabular}{p{1.4cm} | p{0.1cm} p{0.75cm} | p{0.1cm} p{0.75cm} | p{0.1cm} p{0.75cm} | p{0.1cm} p{0.75cm} | p{0.1cm} p{0.75cm} |p{0.1cm} p{0.75cm} }
        \toprule
        \textbf{RT}     &\textbf{FR}    &\textbf{Req}   &\textbf{FR}  &\textbf{Req}   &\textbf{FR}    &\textbf{Req}   &\textbf{FR}  &\textbf{Req}   &\textbf{FR}    &\textbf{Req}   &\textbf{FR}  &\textbf{Req}\\
        \midrule
        &\multicolumn{2}{c|}{AT-v0}  &\multicolumn{2}{c|}{AT-v1}  &\multicolumn{2}{c|}{AT-v2}  &\multicolumn{2}{c|}{AT-v3}\\
        \midrule
        AT-RT0      & 10    & AT2               & 10    & AT2   & 10    & AT2   & 10    & AT2 \\
        AT-RT1      & 10    & AT2               & 10    & AT2   & 10    & AT2   & 10    & AT2 \\
        AT-RT2      & 10    & AT6c              & 10    & AT6c  & 10    & AT6c  & 10    & AT6c \\
        AT-RT3      & 10    & AT1, AT6a, AT6b   & 10    & AT2   & 10    & AT6a  & 10    & AT6c \\
        AT-RT4      & 10    & AT2               & 10    & AT2   & 10    & AT2   & 10    & AT6c\\
        AT-RT5      & 10    & AT2               & 10    & AT2   & 10    & AT2   & 10    & AT6a\\
        \midrule
        &\multicolumn{2}{c|}{CC-v6.1}  &\multicolumn{2}{c|}{CC-v6.2}  &\multicolumn{2}{c|}{CC-v7.1}  &\multicolumn{2}{c|}{CC-v7.3} &\multicolumn{2}{c|}{CC-v7.4}  &\multicolumn{2}{c}{CC-v7.5}\\
        \midrule
        CC-RT0      & 7     & F1    & 7     & F1    & 10    & F1    & 10    & F1    & 10    & F1    & 10    & F1 \\
        CC-RT1      & 7     & F1    & 7     & F1    & 10    & F1    & 10    & F1    & 10    & F1    & 10    & F1 \\
        CC-RT2      & 1     & F1    & 1     & F1    & 3     & F1    & 10    & F1    & 10    & F1    & 9     & F1 \\
        CC-RT3      & 0     & --    & 0     & --    & 1     & F1    & 7     & F1    & 7     & F1    & 6     & F1\\
        \midrule
        &\multicolumn{2}{c|}{OMM-v0}  &\multicolumn{2}{c|}{OMM-v1}  &\multicolumn{2}{c|}{OMM-v2}  &\multicolumn{2}{c}{OMM-v3}\\
        \midrule
        OMM-RT0     & 0     & --    & 10    & OMM2  & 10    & OMM2  & 10    & OMM1 \\
        OMM-RT1     & 0     & --    & 0     & --    & 0     & --    & 0     & -- \\
        OMM-RT2     & 10    & OMM1  & 10    & OMM1  & 10    & OMM1  & 10    & OMM1 \\
        \bottomrule
    \end{tabular}
    \label{tab:RQ1resultsUR}
\end{table*}

\begin{figure*}
    \centering
    \hfill
    \begin{subfigure}{0.3\textwidth}
        \includegraphics[width=\textwidth]{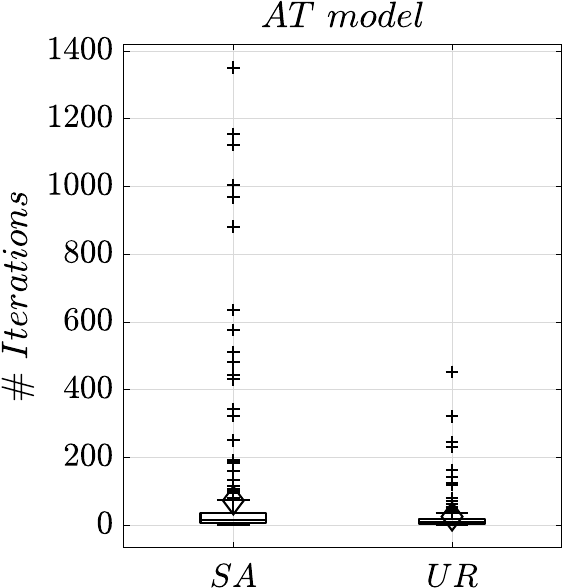}
\end{subfigure}
    \hfill
    \begin{subfigure}{0.3\textwidth}
        \includegraphics[width=\textwidth]{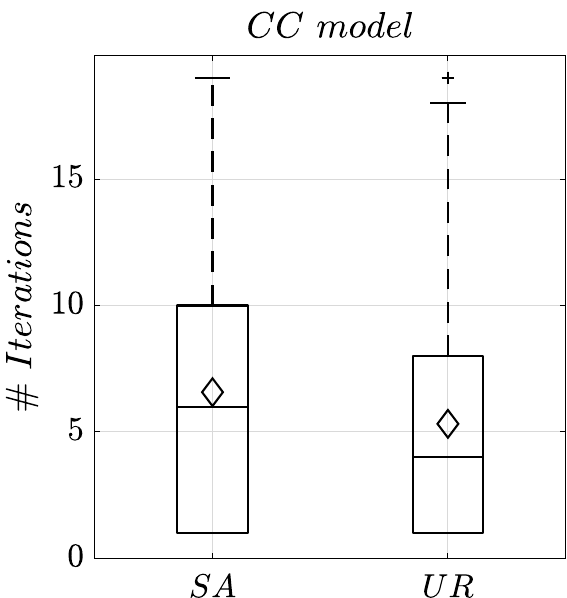}
\end{subfigure}
    \hfill
    \begin{subfigure}{0.3\textwidth}
        \includegraphics[width=\textwidth]{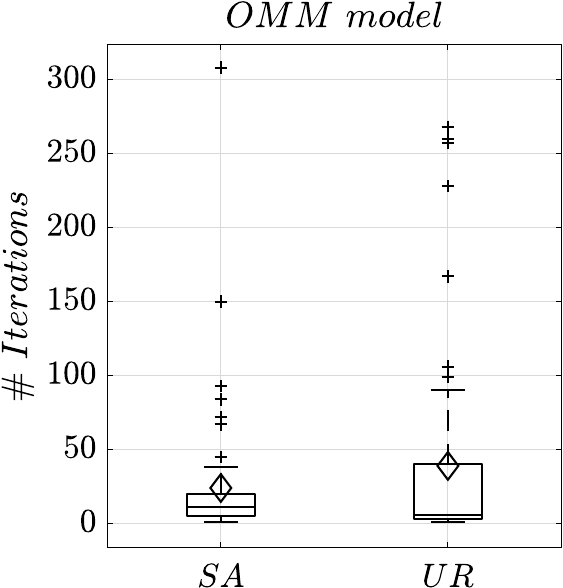}
\end{subfigure}
    \hfill
    \caption{Number of iterations for the failure-revealing runs for both SA and UR. Diamond depicts the average.}
    \label{fig:boxplotsUR}
\end{figure*}

We investigated why UR for AT and CC performed better than SA. 
We visually inspected the behaviors of the algorithms by analyzing the test cases iteratively generated by the two algorithms.
We recorded a video that iteratively plots the inputs generated by the test sequences generated by the test case generation algorithm and the corresponding outputs.
The videos and the list of all the test cases generated by the two algorithms are part of our replication package.
The videos show that the failure-revealing test cases contain profiles with high throttle values for AT.
Indeed, for high throttle values, it is likely that the vehicle (AT1) and the engine (AT2) speed exceed their thresholds (SL1 and RPM2).
For CC, many failure-revealing test cases have the slope signal close to the range boundaries.
This can cause delays in reaching the desired velocity and produce a violation of the Functionality requirement (F1).
Recall that the ranges for the parameters of the Test Sequences for AT and CC are respectively $[5,100]\%$ (Hecate\_throttle1 and Hecate\_throttle2) and $[-4,4]^{\circ}$ (Hecate\_slope).
Based on our qualitative analysis, and a reverse engineering activity we performed on the model, values within the range $[90,100]\%$ for the throttle of AT and $[-4,-2.5]^{\circ}\cup [2.5,4]^{\circ}$ for the slope of CC are likely to lead to a failure-revealing test case.
There is a relatively high chance that UR will generate values within these ranges in a few iterations.
Note that the ARCH competition includes models with different complexity for the generation of failure-revealing test cases (see the corresponding report and the failure-revealing rates of the tools across the different benchmark models) and the effectiveness of UR for the AT model is also confirmed by the last report from the ARCH competition~\cite{Arch2024}.

Therefore, to show the effectiveness of our solution, we restricted the ranges for the input parameters for one model-requirement combination for AT and one for CC.
We kept the search parameters (number of runs and hyperparameters for the two search algorithms) the same. The intuition is that the search-based procedure is useful when the failure is difficult to find; contrarily randomly generating inputs is also an effective approach (and can be more effective than guided search)~\cite{arcuri2011random,iqbal2012empirical}.

For the AT model, we chose the $\langle$AT-v0, AT-RT0$\rangle$ model-requirement combination and we limited the brake signal to the range $[0,150]lb\cdot ft$ and the throttle to the range $[50,85]\%$ and $[50,90]\%$ respectively for \emph{Hecate\_throttle1} and \emph{Hecate\_throttle2}.
The Failure Rate for the UR algorithm was 5 runs out of 10 (changing from 10 runs out of 10 with the original input range), while for the SA algorithm was 10 runs out of 10 (remarkably it remained the same as the one with the original input range).
Furthermore, the average number of iterations in the failure-revealing runs for SA and UR were respectively
147.5 
(min=28,
max=351,
StdDev=95.7)
and 402.8 
(min=109,
max=746,
StdDev=252.6).
This shows how for a more challenging problem, the effectiveness of UR dropped by $50\%$, while the one for SA remained stable, and the number of iterations required increased by over 34 times for UR, while increasing only by 9 times for SA.

For the CC model, we chose the $\langle$CC-v6.2, AT-RT0$\rangle$ model-requirement combination.
We limited the slope signal to the range $[-2.5, 2.5]\mathit{deg}$ and we extended the range for the desired velocity to $[100, 150]\frac{km}{h}$.
We also increased the maximum number of iterations per run from 20 to 50, to give more room for the search algorithms to converge to a solution.
The Failure Rate for the SA algorithm was 8 runs out of 10 (changing from 5 out of 10 with the original input range), while the UR algorithm remained at 7 runs out of 10.
The average number of iterations for SA and UR were respectively
17.1
(min=1,
max=39,
StdDev=13.5)
and 21.3
(min=12,
max=36,
StdDev=8.4).
Therefore, the SA algorithm could find one more failure-revealing test case and it requirede on average 4.2 iterations less than the UR algorithm.

This result shows that our fitness function can effectively drive the
search procedure, i.e., when the failure is difficult to find, it can drive the search toward the failure-revealing test cases.

Further evidence supporting the capability of the fitness function in driving the search is provided by the difference in the set of requirements violations identified by SA and UR. 
For example, the $\langle$AT-v2, AT-RT2$\rangle$ model-RT combination has a Failure Rate of 10 out of 10 runs for both SA and UR, but it achieves this by targeting different requirements.
In this experiment, SA produced a violation of AT1, AT2, AT6a, AT6b, and AT6c respectively in 5 out of 10 runs, 0 out of runs, 5 out of 10 runs, 1 out of 10 runs, and 4 out of 10 runs; while UR produced a violation in 3 out of 10 runs, 0 out of 10 runs, 4 out of 10 runs, 1 out of 10 runs, and 7 out of 10 runs.
Therefore, SA most frequently finds violations in requirements AT1 and AT6a, while UR most frequently finds violations for AT6c.
This shows how the two search algorithms can produce failure-revealing test cases that expose different criticalities in the models under analysis.
To summarize, although comparing the search algorithms is out of scope, these experiments demonstrate that UR and SA are complementary, and our fitness function can effectively drive the search process.

 \section{Related Works}
\label{sec:related}
Our contribution is an SBST approach for Simulink models that supports RT.
Therefore, we consider related works that generate test cases (a)~for Simulink models, (b)~from requirements, and (c)~from tabular requirements specifications since they are the most related to our research.
Our goal is not to present an extensive review of the literature on SBST; The interested reader can consult existing surveys on the topic (e.g., \cite{Panichella_2018,Anand_2013,McMinn_2011,Harman_2012}).

\emph{Simulink Models}.
Many SBST techniques for Simulink models~\cite{Aristeo,formica2022search,Breach,Waga20,FalSTar,falsQBRobCAV2021,FReaK,MoonlightTool,kundu2024data,STaliro,psytalirotool,STGEM} consider requirements specified using logic-based formulae, including approaches for generating automated and online test oracles~\cite{menghi2019generating}.
Several solutions required engineers to manually define fitness functions~\cite{arrieta2019search, Abdessalem_Automotive_2018, Matinnejad_2014, Matinnejad_2016}, or rely on surrogate functions~\cite{nejati2023reflections, Matinnejad_2014, Beglerovic_2017, Wang_2022}.
For example, HECATE~\cite{Hecate2024} computes the fitness from Test Assessment blocks. Unlike these solutions, we target requirements expressed as RT which are significantly different from other artifacts. 

\emph{Requirements}. Many approaches (e.g.,~\cite{STaliro, Breach, FalSTar, FReaK, formica2022search}) generate test cases from requirements.
Recent surveys~\cite{mustafa2021automated, rodrigues2024systematic, Farooq_Tahreem_2022} show that most of the existing techniques focus on UML-based artifacts (e.g.,~\cite{Abdurazik_1999, Bandyopadhyay_2009, Rocha_2021}) and natural language specifications (e.g.,~\cite{Ahsan_2017, Wang_2015, Elghondakly_2015, Carvalho_2013}) confirming the lack of techniques that can generate test cases from requirements expressed in a tabular format.

\emph{Tabular Requirements}. Different techniques~\cite{5975175, peters1998using, 4344109, liu2001generating, clermont2005using} generate test cases from tabular expressions.
Unlike these techniques, our solution targets Simulink models, which are significantly different from other software artifacts as they support signals over time. Some works specifically target the conversion between Stateflow Charts and Function Tables \cite{Singh_2015, Stephen_2019}. 
Our approach is different as it uses the requirement specifications for automatic test cases generation. \section{Conclusion}
\label{sec:conclusion}
This work presented the first SBST approach for Simulink models driven by RT, a tool for specifying software requirements in a tabular format.
We evaluated our solution by considering \modelrtcombinations model-RT combinations.
Our SBST framework returned a failure-revealing test case for \rqonetotal of the combinations. Remarkably, three of these combinations come from a cruise controller model of an industrial simulator that was already analyzed by other tools without finding any violation.
The efficiency of our SBST solution is comparable with other existing SBST tools not based on RT. 
\section*{Data Availability}
The proposed solution, the models and the results are available at \cite{ReplicationPackage}. The replication package will be made publicly available and assigned a DOI upon acceptance.
The CC model cannot be shared in its entirety because it uses Simulink blocks from the VI-CarRealTime library.

\section*{Acknowledgment}
This work was supported in part by the European Union - Next Generation EU. ``Sustainable Mobility Center (Centro Nazionale per la Mobilit\`a Sostenibile - CNMS)'', M4C2 - Investment 1.4, Project Code CN\_00000023, and by project SERICS (PE00000014) under the NRRP MUR program funded by the EU - NGEU.
We acknowledge the support of the Natural Sciences and Engineering Research Council of Canada (NSERC) [funding reference numbers RGPIN-2022-04622, DGECR-2022-0040].\newline
This research was enabled in part by support provided by Compute Ontario (www.computeontario.ca) and Compute Canada (www.computecanada.ca).
 We thank Eduardo Barbosa Louback (McMaster University) for his help in setting up the Hardware-in-the-Loop experiments.

\bibliographystyle{IEEEtran}

\end{document}